\documentclass[11pt]{amsart}
\usepackage{nicefrac}
\usepackage[margin=1.2in]{geometry}
\usepackage{amssymb,amsmath,amsfonts,amsthm}
\usepackage[TS1,T1]{fontenc}    
\usepackage{parskip}
\usepackage{mathabx} 
\usepackage{ifthen,xifthen}
\usepackage[shortlabels]{enumitem}
\usepackage[ngerman,english]{babel}

\usepackage{tikz}
\usepgfmodule{oo}
\usetikzlibrary{arrows,calc,patterns,positioning}
\usepackage[linktocpage,pdftex,final,colorlinks=false,pdfborder={0 0 0}]{hyperref}
\definecolor{darkred}{rgb}{0.5,0,0}
\definecolor{darkgreen}{rgb}{0,0.5,0}
\definecolor{darkblue}{rgb}{0,0,0.5}
\usepackage[capitalise]{cleveref}
\RequirePackage[numbers,square]{natbib}

\numberwithin{equation}{section}    
\usepackage{aliascnt}           
\newcommand{\mynewtheorem}[4][]{
  \ifthenelse{\equal{#1}{}}{    
    \newtheorem{#2}{#3}         
  }{
    \newaliascnt{#2}{#1}        
    \newtheorem{#2}[#2]{#3}     
    \aliascntresetthe{#2}       
  }
  \crefname{#2}{#3}{#4}         
}
\theoremstyle{plain}
\newtheorem{Theorem}{Theorem}[section]
\newtheorem{IntroTheorem}{Theorem}

\newtheorem*{Theorem*}{Theorem}
\crefname{Theorem}{Theorem}{Theorems}
\mynewtheorem[Theorem]{Proposition}{Proposition}{Propositions}
\mynewtheorem[Theorem]{Corollary}{Corollary}{Corollaries}
\mynewtheorem[Theorem]{Lemma}{Lemma}{Lemmas}
\theoremstyle{definition}
\mynewtheorem[Theorem]{Definition}{Definition}{Definitions}
\mynewtheorem[Theorem]{Example}{Example}{Examples}
\theoremstyle{remark}
\mynewtheorem[Theorem]{Remark}{Remark}{Remarks}
\usepackage{wasysym}

\renewcommand{\epsilon}{\varepsilon}
\newcommand{\ve}{\varepsilon}
\renewcommand{\phi}{\varphi}

\DeclareMathOperator{\diam}{diam}

\DeclareMathOperator{\id}{id}

\providecommand{\abs}[1]{\lvert#1\rvert}
\providecommand{\bigabs}[1]{\bigl\lvert#1\bigr\rvert}

\providecommand{\biggabs}[1]{\biggl\lvert#1\biggr\rvert}

\providecommand{\argmt}{\mathchoice{{}\cdot{}}{{}\cdot{}}{{}\bullet{}}{{}\bullet{}}}
\providecommand{\bdry}[1][]{\partial\ifthenelse{\isempty{#1}}{}{^{#1}}}
\providecommand{\setsize}{\abs}

\providecommand{\ceil}[1]{\lceil#1\rceil}

\providecommand{\biggceil}[1]{\biggl\lceil#1\biggr\rceil}

\providecommand{\floor}[1]{\lfloor#1\rfloor}

\providecommand{\norm}[2][]{\lVert#2\rVert\ifthenelse{\equal{}{#1}}{}{_{#1}}}
\providecommand{\bignorm}[2][]{\bigl\lVert#2\bigr\rVert\ifthenelse{\equal{}{#1}}{}{_{#1}}}
\providecommand{\Bignorm}[2][]{\Bigl\lVert#2\Bigr\rVert\ifthenelse{\equal{}{#1}}{}{_{#1}}}
\providecommand{\biggnorm}[2][]{\biggl\lVert#2\biggr\rVert\ifthenelse{\equal{}{#1}}{}{_{#1}}}
\providecommand{\Biggnorm}[2][]{\Biggl\lVert#2\Biggr\rVert\ifthenelse{\equal{}{#1}}{}{_{#1}}}
\providecommand{\Norm}[2][]{\left\lVert#2\right\rVert\ifthenelse{\equal{}{#1}}{}{_{#1}}}

\providecommand{\spr}[3][]{{\langle#2,#3\rangle}\ifthenelse{\equal{}{#1}}{}{_{#1}}}
\providecommand{\Spr}[3][]{\left\langle#2,#3\right\rangle\ifthenelse{\equal{}{#1}}{}{_{#1}}}
\providecommand{\from}{\colon}
\providecommand{\unitmass}[1]{\delta_{#1}}
\providecommand{\xto}{\xrightarrow}
\providecommand{\qtext}{\quad\text}

\providecommand{\textq}[1]{\text{#1}\quad}
\providecommand{\qtextq}[1]{\quad\text{#1}\quad}
\providecommand{\Ioo}[3][]{\mathopen(#2,#3\mathclose)\ifthenelse{\equal{}{#1}}{}{_{#1}}}
\providecommand{\Ico}[3][]{\mathopen[#2,#3\mathclose)\ifthenelse{\equal{}{#1}}{}{_{#1}}}
\providecommand{\Ioc}[3][]{\mathopen(#2,#3\mathclose]\ifthenelse{\equal{}{#1}}{}{_{#1}}}
\providecommand{\Icc}[3][]{\mathopen[#2,#3\mathclose]\ifthenelse{\equal{}{#1}}{}{_{#1}}}
\providecommand{\ifu}[1]{\mathbf1\ifthenelse{\equal{#1}{}}{}{_{#1}}}
\providecommand{\dnto}{\searrow}%
\providecommand{\isect}{\cap}
\providecommand{\Isect}{\bigcap}
\providecommand{\union}{\cup}
\providecommand{\Union}{\bigcup}
\providecommand{\dunion}{\mathop{\dot\union}}

\providecommand{\tensor}{\otimes}
\providecommand{\Tensor}{\bigotimes}
\providecommand{\dirac}[1]{\delta_{#1}}

\providecommand{\constf}[1][f]{K_{#1}}
\providecommand{\constb}[1][f]{D_{#1}}
\let\originald\d 
\renewcommand{\d}{\ifthenelse{\boolean{mmode}}{\mathrm d}{\originald}}
\providecommand{\dd}{\,\d}
\def\Int#1d{\int#1\dd} 

\newcommand{\B}{\mathbb{B}}

\newcommand{\E}{\mathbb{E}}
\newcommand{\N}{\mathbb{N}}
\let\originalP\P 
\renewcommand{\P}{\ifthenelse{\boolean{mmode}}{\mathbb P}{\originalP}}

\newcommand{\R}{\mathbb{R}}
\newcommand{\Z}{\mathbb{Z}}

\newcommand{\cA}{\mathcal{A}}
\newcommand{\cB}{\mathcal{B}}
\providecommand{\Borel}{\cB}

\newcommand{\cD}{\mathcal{D}}

\newcommand{\cF}{\mathcal{F}}

\newcommand{\cM}{\mathcal{M}}

\newcommand{\cP}{\mathcal{P}}

\newcommand{\cS}{\mathcal{S}}
\newcommand{\cT}{\mathcal{T}}
\newcommand{\cU}{\mathcal{U}}

\newcommand{\cY}{\mathcal{Y}}

\providecommand{\given}{\mid}

\providecommand{\superOmega}{\underline{\Omega}}%
\providecommand{\superomega}{\underline{\omega}}%
\providecommand{\superP}{\underline{\P}}%
\providecommand{\superE}{\underline{\E}}%
\providecommand{\superL}{\underline{L}}%
\providecommand{\superA}{\underline{\cA}}%
\providecommand{\superX}{X_0}%

\newcommand{\hm}[1]{\textbf{*}\leavevmode{\marginpar{\tiny%
$\hbox to 0mm{\hspace*{-0.5mm}$\leftarrow$\hss}%
\vcenter{\vrule depth 0.1mm height 0.1mm width \the\marginparwidth}%
\hbox to 0mm{\hss$\rightarrow$\hspace*{-0.5mm}}$\\\relax\raggedright #1}}}
\usepackage[final]{microtype}
\usepackage{lmodern}
\usepackage[textsize=footnotesize,colorinlistoftodos]{todonotes}

\title[Glivenko--Cantelli Theory for almost additive fields on amenable groups]
{Glivenko--Cantelli Theory, Ornstein--Weiss quasi-tilings, and uniform Ergodic Theorems
for distribution-valued fields over amenable groups}
\author{Christoph Schumacher}
\address[CS \& IV]{Fakult\"at f\"ur Mathematik, TU Dortmund, 44221 Dortmund, Germany}
\author{Fabian Schwarzenberger}
\address[FS]{Fakult\"at f\"ur Informatik/Mathematik, HTW Dresden, 01069 Dresden, Germany}
\author{Ivan Veseli\'c}
\thanks{
\today, \jobname.tex}

\begin{document}

\maketitle


\begin{abstract}
We consider random fields indexed by finite subsets of an amenable
discrete group, taking values in the Banach-space of bounded
right-continuous functions. The field is assumed to be equivariant,
local, coordinate-wise monotone, and almost additive, with finite range
dependence. Using the theory of quasi-tilings we prove an uniform ergodic
theorem, more precisely, that averages along a Foelner sequence converge
uniformly to a limiting function. Moreover we give explicit error
estimates for the approximation in the sup norm.
\end{abstract}

\section{Introduction}

Ergodic theorems for Banach space valued functions or fields
have been studied among others in \cite{LenzMV-08,LenzSV-10,PogorzelskiS-16}
in a combinatorial setting.
The three quoted papers consider different group actions
in increasing generality:
the lattice~$\Z^d$, monotilable amenable discrete groups
and general amenable discrete groups, respectively.
Note that amenability is a natural assumption
for the validity of the ergodic theorem, as shown explicitly in
\cite{Tao-15}. Already before that combinatorial ergodic theorems for Banach space valued functions
have been proven in the context of Delone dynamical systems,
see \cite{LenzS-05} and the references therein.

The combinatorial framework offers the advantage of a minimum of probabilistic or measure theoretic assumptions,
the necessary one being that frequencies or densities of finite patterns are well defined and can be approximated by an exhaustion
(corresponding to a law of large numbers).
A disadvantage of the combinatorial approach chosen,
is that the range of colours
(or the alphabet corresponding to the values of the random variables)
needs to be finite.
Also, the derived ergodic theorems are in a sense conditional:
The convergence bound depends on the speed of convergence of the pattern frequencies.

Our present research aims at dispensing with the finiteness condition on the set of colours.
The price to pay is that we have to assume more probabilistic structure and in particular independence or at least finite range correlations.
In return, this structure yields automatically quantitative approximation error bounds.
No extra assumptions on the speed of convergence of the pattern frequencies is needed.
For the case of fields defined over $\Z^d$ and $\Z^d$-actions
we have established such an ergodic theorem in \cite{SchumacherSV-16},
which takes on the form of a Glivenko--Cantelli theorem,
and  which we recall now in an informal way.

\begin{IntroTheorem}[\cite{SchumacherSV-16}]
  Let $\Lambda_n=\Ico0n^d \isect \Z^d$, and
  $\omega = (\omega_g)_{g \in \Z^d} \in \R^{\Z^d}$
  be an i.\,i.\,d.\ sequence of real random variables.
  Assume the field
  \begin{equation*}
    f\from\cP(\Z^d)\times\R^{\Z^d}\to\B:=\{D\from\R\to\R\mid
    \text{$D$ right-continuous and bounded}\}
  \end{equation*}
  is $\Z^d$-equivariant, monotone in each coordinate~$\omega_g$,
  local, and almost additive, i.\,e.\ for disjoint
  $Q_1,\dots,Q_n\subseteq \Z^d$ and $Q:=\bigcup_{i=1}^nQ_i$ we have
  \[
    \Bignorm{f(Q,\omega)-\sum_{i=1}^nf(Q_i,\omega)}_\infty
    \le\sum_{i=1}^n\bigabs{\partial Q_i}\text,
  \]
  where $\partial Q_i$ denotes the boundary set.
  Assume furthermore, that
  $f_\infty:=\sup_\omega\norm[\infty]{f(\id ,\omega)}<\infty$.

  Then there is a function $f^*\colon \R\to\R$ such that for each $m\in\N$,
  there exist $a(m),b(m)>0$, such that for all $j\in\N$, $j>2m$,
  there is an event $\Omega_{j,m}\subseteq \R^{\Z^d}$, with the properties
  \begin{equation*}
    \P(\Omega_{j,m})
    \ge1-b(m)\exp\bigl(-a(m)\setsize{\Lambda_j}\bigr)
  \end{equation*}
  and
  \begin{equation*}
    \forall \omega \in \Omega_{j,m}\colon \quad
    \biggnorm[\infty]{\frac{f(\Lambda_j,\omega)}{\setsize{\Lambda_j}}- f^*}
      \le2^{2d+1}\Bigl(\frac{(6d+3+2f_\infty)m^d+1}{j-2m}
        +\frac4m\Bigr)
  \end{equation*}
  In particular, almost surely we have
  $\lim\limits_{n\to\infty}
  \bignorm{\frac{f(\Lambda_n,\argmt)}{\setsize{\Lambda_n}}-f^*}_\infty
  =0$.
\end{IntroTheorem}
For a precise formulation of the properties of the field~$f$
see Section~\ref{sec:results}.
Let us note that in our Theorem~$f$ takes values in the Banach space~$\B$
of right continuous and bounded functions with $\sup$-norm
while in \cite{LenzMV-08,LenzSV-10,PogorzelskiS-16}
an arbitrary Banach space was allowed.
This restriction is due to our use of the Glivenko--Cantelli theory in the
proof and currently we do not know how to extend it to arbitrary Banach spaces.

Naturally one asks whether the above result and its proof extend to general finitely generated amenable groups.
In this case, obviously, the boundary has to be taken with respect to a generating set $S\subseteq G$,
and the sequence of squares~$\Lambda_n$ has to be replaced  by a F{\o}lner sequence.
Indeed, if~$G$ satisfies additionally
\begin{enumerate}[($\boxplus$)]
  \item\label{tiling}
    There exists  a F{\o}lner sequence $(\Lambda_n)_{n\in\N}$ in~$G$,
    and a sequence of symmetric  \emph{grids} $T_n=T_n^{-1}\subseteq G$
    such that     $G=\dot\bigcup_{t\in T_n}\Lambda_nt$  is a disjoint union.
\end{enumerate}
the proofs of \cite{SchumacherSV-16}
apply with technical, but no strategic, modifications, as sketched in Appendix \ref{monotile}.

However, it is \emph{not clear} in which generality assumption \labelcref{tiling} holds.
In fact, the existence of tiling F{\o}lner sequences (for general amenable groups) has been investigated in several instances.
It turned out that there exist useful additional conditions which imply the validity of \labelcref{tiling}, cf.~\cite{Weiss-01, Krieger}.
For instance, a group which is residually finite and amenable contains a tiling F{\o}lner sequence.
Unfortunately, there is a lack of the complete picture: It is still an open question whether there exists a tiling F{\o}lner sequence in each amenable group.

Since this question seems hard to answer, Ornstein and Weiss invented in \cite{OrnsteinW-87} the theory of $\epsilon$-quasi tilings.
The idea is to consider a tiling which is in several senses weaker as the one in \labelcref{tiling}.
For a given $\epsilon>0$ one has the following properties:
\begin{itemize}
  \item the group is not tiled with one element of a F{\o}lner sequence, but with finitely many elements of this sequence; the number of these elements depends on~$\epsilon$;
  \item the tiles are allowed to overlap, but the proportion of the part of any tile which is allowed to intersect other tiles is at most of size~$\epsilon$.
    This property is called $\epsilon$-disjointness;
  \item each element of a F{\o}lner sequence with a sufficiently large index is, up to a proportion of size $\epsilon$ the union of $\epsilon$-disjoint tiles.
\end{itemize}
The authors showed that each amenable group can be $\epsilon$-quasi tiled.
In \cite{PogorzelskiS-16} these ideas have been developed further
in order to obtain quantitative estimates on the portion
which is covered by translates of one specific element of the tiles.
The proof of our main result, which we state now in an informal way,
is based on these results on quasi tilings.

\begin{IntroTheorem}
  Let~$(\Lambda_n)$ be a F{\o}lner sequence in a finitely generated
  group~$G$. Let $\omega=(\omega_g)_{g\in G}\in\R^G$
  be an i.\,i.\,d.\ sequence of real random variables.
  Assume the field
  \[
    f \from\cP(G)\times\R^G\to\{D\from\R\to\R\mid\text{$D$ right-continuous and bounded}\},
  \]
  is $G$-equivariant, monotone in each coordinate $\omega_g$, local, and almost additive, i.\,e.\ for disjoint
  $Q_1,\dots,Q_n\subseteq G$ and $Q:=\bigcup_{i=1}^nQ_i$ we have
  \[
    \Bignorm{f(Q,\omega)-\sum_{i=1}^nf(Q_i,\omega)}_\infty
    \le\sum_{i=1}^n\bigabs{\partial Q_i}\text,
  \]
  where $\partial Q_i$ denotes the boundary relative to a set of generators $S\subseteq G$.
  Assume furthermore, that $f_\infty:=\sup_\omega\norm[\infty]{f(\id ,\omega)}<\infty$.

  Then there is a function $f^*\colon \R\to\R$ such that for each $\delta\in\Ioo0{1}$,
  there exist $a(\delta)>0$, such that for all sufficiently large $j\in \N$,
  there is an event $\Omega_{j,\delta}\subseteq \R^G$, with the properties
  \begin{equation*}
    \P(\Omega_{j,\delta})
    \ge1-     \exp\bigl(-a(\delta)\setsize{\Lambda_j}\bigr)
  \end{equation*}
  and
  \begin{align*}
    \forall \omega \in \Omega_{j,\delta} \quad \quad
    \biggnorm{\frac{f(\Lambda_j,\omega)}{\setsize{\Lambda_j}}-f^*}_\infty&
    \le(37 f_\infty+84\setsize S+131)\delta\text.
  \end{align*}
  In particluar,  almost surely we have $
  \lim\limits_{n\to\infty}
  \Bignorm[\infty]{\frac{f(\Lambda_n,\argmt)}{\setsize{\Lambda_n}}-f^*}
  =0$.
\end{IntroTheorem}
For a precise formulation, see Definition~\ref{def:admissible} and Theorem \ref{thm:main}.
To achieve the error bound in the theorem, we work with an $\ve$-quasi tiling with $\ve=\delta^2$.

\begin{Remark}
 Let us sketch the difference between the proof of Theorem B (see also \cref{thm:main} below)
  and the Theorem  2.8 of \cite{SchumacherSV-16} sketched as Theorem A above.
  There we heavily relied on the fact that~$\Z^d$ can be tiled exactly with any cube of integer length.
  Since a general discrete amenable group need not have such a tiling,
  we have to modify the \emph{geometric parts} of the proof
  and use $\ve$-quasi tilings as in \cite{OrnsteinW-87,PogorzelskiS-16}.
  Since quasi tilings in general overlap,
  we loose independence of the corresponding random variables.
  This requires a change in the \emph{probabilistic part}
  of the proof and in particular the use of resampling.

\end{Remark}


The structure of the paper is as follows.
In \cref{sec:results}, we precisely describe the model and our result.
In \cref{sec:quasitilings} we summarize results about $\ve$-quasi tilings, which are fundamental for our proof.
The error estimate in the main theorem and the corresponding approximation procedure naturally split in three parts,
which are treated consecutively in \cref{secboundary,sec:GC,sec:cauchy}.
\Cref{secboundary} is of geometric nature.
\Cref{sec:GC} is based on multivariate Glivenko--Cantelli theory.
\Cref{sec:cauchy} is geometric in spirit again.
In the Appendix  we prove a resampling lemma and
indicate how the proof of \cite{SchumacherSV-16} could be adapted to cover monotileable amenable groups.


\section{Model and main results}\label{sec:results}
We start this section with the introduction of the geometric and probabilistic setting:
We recall the notion of a Cayley graph of an amenable group~$G$,
introduce random colorings of vertices,
and define so-called admissible fields,
which are random functions mapping finite subsets of~$G$
to functions on~$\R$ and satisfying a number of natural properties,
cf.~Definition~\ref{def:admissible}.
We are then in the position to formulate our main Theorem~\ref{thm:main}.

Let~$G$ be a finitely generated group and $S=S^{-1}\subseteq G\setminus\{\id\}$ a finite generating system.
Obviously~$G$ is countable.
The set of all finite subsets of~$G$ is denoted by~$\cF$ and is countable as well.
Throughout this paper we will assume that~$G$ is amenable,
i.\,e.\ there exists a squence $(\Lambda_n)_{n\in\N}$ of elements in~$\cF$
such that for each $K\in\cF$ one has
\begin{equation}\label{eq:folner}
  \frac{\setsize{\Lambda_n\triangle K\Lambda_n}}{\setsize{\Lambda_n}}
  \xto{n\to\infty}0\text.
\end{equation}
Here, $K\Lambda_n:=\{kg\mid k\in K,g\in\Lambda_n\}$
is the pointwise group multiplication of sets,
$\Lambda_n\triangle K\Lambda_n$
denotes the symmetric difference between the sets~$\Lambda_n$ and~$K\Lambda_n$,
and $\setsize{A}$ denotes the cardinality of the finite set~$A$.
A sequence $(\Lambda_n)_{n\in\N}$ satisfying property~\eqref{eq:folner}
is called \emph{F{\o}lner sequence}.

The pair $(G,S)$ gives rise to an undirected graph $\Gamma(G,S)=(V,E)$
with vertex set $V:=G$ and edge set
$  E:=\{\{x,y\} \mid xy^{-1}\in S\}$.
The graph $\Gamma(G,S)$ is known as the \emph{Cayley graph} of~$G$
with respect to the generating system~$S$.
Note that by symmetry of~$S$ the edge set~$E$ is well-defined.
Let $d\from G\times G\to\N_0$ denote the usual graph metric of $\Gamma(G,S)$.
The distance between two non-empty sets $\Lambda_1,\Lambda_2\subseteq G$
is given by
\begin{equation*}
  d(\Lambda_1,\Lambda_2)
  :=\min\{d(x,y)\mid x\in \Lambda_1,y\in\Lambda_2\}\text.
\end{equation*}
In the case where $\Lambda_1=\{x\}$ consists of only one element, we write
$d(x,\Lambda_2)$ for $d(\{x\},\Lambda_2)$.
The diameter of a non-empty set $\Lambda\in\cF$
is defined by $\diam(\Lambda):=\max\{d(x,y)\mid x,y\in\Lambda\}$.

Given $r\ge0$, the \emph{$r$-boundary} of a set $\Lambda\subseteq G$
is defined by
\begin{equation*}
  \bdry[r](\Lambda)
  :=\{x\in \Lambda\mid d(x,G\setminus\Lambda)\le r\}
    \union\{x\in G\setminus\Lambda\mid d(x,\Lambda)\le r\}\text.
\end{equation*}
and besides this we use the notation
\begin{equation*}
  \Lambda^r:=\Lambda\setminus\bdry[r](\Lambda)
  =\{x\in\Lambda\mid d(x,G\setminus\Lambda)>r\}\text.
\end{equation*}
It is easy to verify that for a given F{\o}lner sequence $(\Lambda_n)_{n\in\N}$,
or $(\Lambda_n)$ for short, and $r\ge0$ we have
\begin{equation}\label{eq:folner2}
  \lim_{n\to\infty}\frac{\setsize{\bdry[r](\Lambda_n)}}{\setsize{\Lambda_n}}=0
\quad\text{and}\quad
  \lim_{n\to\infty} \frac{\setsize{\Lambda_n^r}}{\setsize{\Lambda_n}}=1\text.
\end{equation}
Moreover, if $(\Lambda_n)$ is a F{\o}lner sequence, then for arbitrary $r\ge0$
the sequence $(\Lambda_n^r)$ is a F{\o}lner sequence as well.
Conversely, in order to show that a given sequence $(\Lambda_n)$
is a F{\o}lner sequence, it is sufficient \cite{Adachi-93,diss-fabian}
to show for $n\to\infty$ either
\begin{align}\label{eq:folner:suff}
  \frac{\setsize{\Lambda_n \triangle S\Lambda_n}}{\setsize{\Lambda_n}}\to0
  \qquad\text{or}\qquad
  \frac{\setsize{\bdry[1](\Lambda_n)}}{\setsize{\Lambda_n}}\to0
  \text.
\end{align}

Let us introduce colorings of the group~$G$
(or equivalently colorings of the vertices of $\Gamma(G,S)$).
We choose a (finite or infinite) set of possible colors $\cA\in\Borel(\R)$.
The sample set,
\begin{equation*}
  \Omega=\cA^{G}=\{\omega=(\omega_g)_{g\in G}\mid\omega_j\in\cA\}\text,
\end{equation*}
is the set of all possible colorings of~$G$.
Note that~$G$ acts in a natural way via translations on~$\Omega$.
To be precise, we define for each $g\in G$
\begin{equation}\label{eq:def:tau}
  \tau_g\from\Omega\to\Omega\textq,
  (\tau_g\omega)_x=\omega_{xg}\textq,(x\in G)\text.
\end{equation}
Next, we introduce random colorings.
As the $\sigma$-algebra we choose~$\cB(\Omega)$,
the product $\sigma$-algebra on~$\Omega$ generated by cylinder sets.
Oftentimes, we are interested in (finite) products of~$\cA$
embedded in the infinite product space~$\Omega$.
To this end, we set for $\Lambda\subseteq G$
\begin{equation*}
  \Omega_\Lambda
  :=\cA^\Lambda
  :=\{(\omega_g)_{g\in \Lambda}\mid \omega_g\in\cA\}
\end{equation*}
and define
\begin{equation*}
  \Pi_\Lambda\from\Omega\to \Omega_\Lambda\qquad \text{by}\qquad
  (\Pi_\Lambda (\omega))_g := \omega_g\qtextq{ for each} g\in \Lambda.
\end{equation*}
As shorthand notation we write~$\omega_\Lambda$
instead of $\Pi_\Lambda(\omega)$.
Having introduced the measurable space $(\Omega,\cB(\Omega))$,
we choose a probability measure~$\P$ with the following properties:
\begin{enumerate}[(M1)]
  \item\label{M1} \emph{equivariance:} 
    For each $g\in G$ we have $\P\circ\tau_g^{-1}=\P$.
  \item\label{M2} \emph{existence of densities:}
    There is a $\sigma$-finite measure $\mu_0$ on~$(\cA,\cB(\cA))$,
    such that for each $\Lambda\in\cF$ the measure
    $\P_\Lambda:=\P\circ\Pi_{\Lambda}^{-1}$
    is absolutely continuous with respect to
    $\mu_\Lambda:=\Tensor_{g\in\Lambda}\mu_0$ on~$\Omega_\Lambda$.
    We denote the corresponding probability density function by~$\rho_\Lambda$.
  \item\label{M3} \emph{independence condition:}
    There exists $r\geq 0$ such that for all $n\in\N$
    and non-empty $\Lambda_1,\dots,\Lambda_n\in\cF$
    with $\min\{d(\Lambda_i,\Lambda_j)\mid i\neq j\}> r$
    we have $\rho_{\Lambda} =\prod_{j=1}^n \rho_{\Lambda_j}$,
    where $\Lambda=\bigcup_{j=1}^n \Lambda_j$.
\end{enumerate}
The measure~$\P_\Lambda$ is called the \emph{marginal measure} of~$\P$.
It is defined on $(\Omega_\Lambda,\cB(\Omega_\Lambda))$,
where again $\cB(\Omega_\Lambda)$
is generated by the corresponding cylinder sets.

\begin{Remark}
(a) \ The constant $r\ge0$ in~\ref{M3}
    can be interpreted as the correlation length.
    In particular, if $r=0$ this property implies
    that the colors of the vertices are chosen independently.\\
(b) \ \ref{M2} is trivially satisfied, if~$\P$ is a product measure.
\end{Remark}

In the following, we consider partial orderings on~$\Omega$ and on~$\R^k$,
respectively.
Here we write $\omega\le\omega'$ for $\omega,\omega'\in\Omega$,
if for all $g\in G$ we have $\omega_g\le\omega'_g$.
The notion $x\leq x'$ for $x,x'\in\R^k$ is defined in the same way.
We consider the Banach space
\begin{equation*}
  \B:=\{F\from\R\to\R\mid\text{$F$ right-continuous and bounded}\}\text,
\end{equation*}
which is equipped with supremum norm~$\norm{\argmt}:=\norm[\infty]{\argmt}$.

\begin{Definition}\label{def:admissible}
  A field $f \from\cF\times\Omega\to\B$ is called \emph{admissible}
  if the following conditions are satisfied
  \begin{enumerate}[({A}1)]
  \item\label{transitive}%
    equivariance: 
    for $\Lambda\in\cF$, $g\in G$ and $\omega\in\Omega$ we have
    \begin{equation*}
      f(\Lambda g,\omega)=f(\Lambda,\tau_g\omega)\text.
    \end{equation*}
  \item\label{local}%
    locality: for all $\Lambda\in\cF$ and $\omega,\omega'\in\Omega$
    satisfying $\Pi_\Lambda(\omega)=\Pi_\Lambda(\omega')$ we have
    \begin{equation*}
      f(\Lambda,\omega)=f(\Lambda,\omega')\text.
    \end{equation*}
   \item\label{extensive}%
    almost additivity:
    for arbitrary $\omega\in\Omega$,
    pairwise disjoint $\Lambda_1,\dots,\Lambda_n\in\cF$ and
    $\Lambda:=\bigcup_{i=1}^n\Lambda_i$ we have
    \[
      \Bignorm{f(\Lambda,\omega)-\sum_{i=1}^nf(\Lambda_i,\omega)}
        \le\sum_{i=1}^nb(\Lambda_i)\text,
    \]
    where $b\from\cF\to\Ico0\infty$ satisfies
    \begin{itemize}
      \item $b(\Lambda)=b(\Lambda g)$
        for arbitrary $\Lambda\in\cF$ and $g\in G$,
      \item $\exists\constb>0$ with $b(\Lambda)\le\constb\setsize\Lambda$
        for arbitrary $\Lambda\in\cF$,
      \item $\lim_{i\to\infty}b(\Lambda_i)/\setsize{\Lambda_i}=0$,
        if $(\Lambda_i)_{i\in\N}$ is a F{\o}lner sequence.
      \item for $\Lambda,\Lambda'\in\cF$ we have
        $b(\Lambda\cup \Lambda')\le b(\Lambda)+b(\Lambda')$,
        $b(\Lambda\cap \Lambda')\le b(\Lambda)+b(\Lambda')$, and
        $b(\Lambda\setminus\Lambda')\le b(\Lambda)+b(\Lambda')$.
    \end{itemize}
  \item\label{monotone}%
    monotonicity:~$f$ is antitone with respect to the partial orderings on
    $\Omega\subseteq\R^{G}$ and~$\B$,
    i.\,e.\ if $\omega,\omega'\in \Omega$ satisfy $\omega\le\omega'$, we have
    \[
      f(\Lambda,\omega)(x)\geq  f(\Lambda,\omega')(x)
      \qtext{for all $x\in\R$ and $\Lambda\in\cF$.}
    \]
  \item\label{bounded}
    boundedness:
    \begin{equation*}
      \sup_{\omega\in\Omega}\norm{f(\{\id\},\omega)}
      <\infty\text.
    \end{equation*}
  \end{enumerate}
\end{Definition}

\begin{Remark}\label{rem:admissible}
  \begin{itemize}
  \item Locality \ref{local} can be formulated as follows:
    $f(\Lambda,\argmt)$ is $\sigma(\Pi_\Lambda)$-measurable.
    This enables us to define $f_\Lambda\from\Omega_\Lambda\to\B$ by
    $f_\Lambda(\omega_\Lambda):=f(\Lambda,\omega)$ with $\Lambda\in\cF$
    and $\omega\in\Omega$.
  \item We call the function~$b$ in \ref{extensive} \emph{boundary term}.
    Note that the fourth assumption on~$b$ in~\ref{extensive}
    was not made in \cite{SchumacherSV-16}.
    Indeed, this inequality is used to separate overlapping tiles
    and is unnecessary as soon as the group has the tiling property~\labelcref{tiling}.
    This fourth point is used only in \cref{la:quasialmostadd,lemma:insertindependence}.
  \item The antitonicity assumption in \ref{monotone} can be weakend.
    In particular, our proofs apply to fields
    which are monotone in each coordinate, where the direction
    of the monotonicity can be different for distinct coordinates.
    For simplicity reasons and as our main example (see \cite{SchumacherSV-16})
    satisfies \ref{monotone},
    we restrict ourselves to this kind of monotonicity.
  \item As shown in \cite{SchumacherSV-16}, a combination of
    \labelcref{transitive,extensive,bounded}
    implies that the bound
    \begin{align}\label{eq:bound-K_f}
      \constf:=\sup\{\norm{f(\Lambda,\omega)}/\setsize\Lambda
      \mid\omega\in\Omega,\Lambda\in\cF\}
      \le\constb+\sup_{\omega\in\Omega}\norm{f(\{\id\},\omega)}<\infty\text.
    \end{align}
  \end{itemize}
\end{Remark}

\begin{Definition}
  A set~$\cU$ of admissible fields is called \emph{admissible set},
  if their bound is uniform:
  \begin{equation*}
    \constf[\cU]:=\sup_{f\in\cU}\constf<\infty
  \end{equation*}
  and each for each $f\in\cU$ condition \ref{extensive} is satisfied with the same boundary term~$b$.
  In this situation we denote the constant in \ref{extensive} by $D_{\cU}$.
\end{Definition}

Let us state the main theorem of this paper.
\begin{Theorem}\label{thm:main}
  Let~$G$ be a finitely generated amenable group with a F{\o}lner sequence  $(\Lambda_n)$.
  Further, let~$\cA\in \cB(\R)$ and $(\Omega=\cA^G,\cB(\Omega),\P)$
  a probability space such that~$\P$ satisfies \labelcref{M1,M2,M3}.
  Finally, let~$\cU$ be an admissible set.
  \par
  \begin{enumerate}[(a), wide]
  \item Then, there exists an event $\tilde\Omega\in\cB(\Omega)$
    such that $\P(\tilde\Omega)=1$ and for any $f\in\cU$
    there exists a function $f^*\in\B$, which does not depend on the specific F{\o}lner seqeunce $(\Lambda_n)$, with
    \begin{equation*}
      \forall\omega\in\tilde\Omega\colon\quad
      \lim_{n\to\infty}
      \biggnorm{\frac{f(\Lambda_n,\omega)}{\setsize{\Lambda_n}}-f^*}
      =0
      \text.
    \end{equation*}
  \item Furthermore, for each $\ve\in\Ioo0{1/10}$,
    there exist $j_0(\ve)\in\N$, independent of~$\constf[\cU]$,
    and $a(\ve,\constf[\cU]),b(\ve,\constf[\cU])>0$,
    such that for all $j\in\N$, $j\ge j_0(\ve)$,
    there is an event $\Omega_{j,\ve,\constf[\cU]}\in\Borel(\Omega)$,
    with the properties
    \begin{equation*}
      \P(\Omega_{j,\ve,\constf[\cU]})
      \ge1-b(\ve,\constf[\cU])
      \exp\bigl(-a(\ve,\constf[\cU])\setsize{\Lambda_j}\bigr)
    \end{equation*}
    and
    \begin{align*}
      \biggnorm{\frac{f(\Lambda_j,\omega)}{\setsize{\Lambda_j}}-f^*}&
      \le(37\constf[\cU]+47\constb[\cU]+47)\sqrt\ve
      \qtext{ for all $\omega\in\Omega_{j,\ve,\constf[\cU]}$ and all $f\in\cU$.}
    \end{align*}
  \end{enumerate}
\end{Theorem}

For examples of measures~$\P$ satisfying \labelcref{M1,M2,M3}
and of admissible fields, we refer to \cite{SchumacherSV-16}.
The generalization of the geometry from the lattice $\Z^d$ to an  amenable group~$G$
does not affect the examples.
See also \cite{Veselic2005,LenzVeselic2009} for a discussion
of models giving rise to a discontinuous integrated density of states,
which nevertheless can be uniformly approximated by almost additive fields.

\section{Outline of $\epsilon$-quasi tilings}\label{sec:quasitilings}

Let us give a brief introduction to the theory of $\epsilon$-quasi tilings.
The main ideas go back to Ornstein and Weiss in \cite{OrnsteinW-87}.
However the specific results we use here are taken from \cite{PogorzelskiS-16},
see also \cite{diss-fabian}.

Let $(Q_n)$ be a F{\o}lner sequence.
This sequence is called \emph{nested}, if for all $n\in\N$
we have $\{\id\}\subseteq Q_n\subseteq Q_{n+1}$.
Using tranlations and subsequences it is easy to show
that every amenable group contains a nested F{\o}lner sequence,
c.\,f.\ \cite[Lemma 2.6]{PogorzelskiS-16}.

We will use the elements of the nested F{\o}lner sequence~$(Q_n)$ to
$\epsilon$-quasi tile elements of a given F{\o}lner sequence~$(\Lambda_j)$
for (very) large index~$j$.
The next definition provides the notion of an $\alpha$-covering,
$\epsilon$-disjointness, and $\epsilon$-quasi tiling.

\begin{Definition}\label{def:qt}
  Let~$G$ be a finitely generated group,
	$\alpha,\ve\in\Ioo01$, and~$I$ some index set.
	\begin{itemize}
	\item The sets $Q_i\in\cF$, $i\in I$, are said to \emph{$\alpha$-cover}
		the set~$\Lambda\in\cF$, if
		\begin{enumerate}[(i),noitemsep]
		\item $\Union_{i\in I}Q_i\subseteq\Lambda$, and
		\item $\setsize{\Lambda\isect\Union_{i\in I}Q_i}
						\ge\alpha\setsize{\Lambda}$.
		\end{enumerate}
	\item The sets $Q_i\in\cF$, $i\in I$, are called \emph{$\epsilon$-disjoint},
		if there are subsets $\mathring{Q_i}\subseteq Q_i$, $i\in I$,
		such that for all $i\in I$ we have
		\begin{enumerate}[(i),noitemsep]
		\item $\setsize{Q_i\setminus\mathring Q_i}\le\epsilon\setsize{Q_i}$, and
    \item $\mathring Q_i$ and $\Union_{j\in I\setminus\{i\}}\mathring Q_j$
      are disjoint.
		\end{enumerate}
	\item The $K_i\in\cF$, $i\in I$, are said to
		\emph{$\ve$-quasi tile~$\Lambda\in\cF$}, if
		there exist $T_i\in\cF$, $i\in I$, such that
		\begin{enumerate}[(i),noitemsep]
		\item the elements of $\{K_iT_i\mid i\in I\}$ are pairwise disjoint,
		\item for each $i\in I$,
			the elements of $\{K_it\mid t\in T_i\}$ are $\ve$-disjoint, and
		\item\label{qt:2epscover}
			the family $\{K_iT_i\mid i\in I\}$ $(1-2\ve)$-covers~$\Lambda$.
		\end{enumerate}
		The set~$T_i$ is called \emph{center set} for the \emph{tile~$K_i$},
		$i\in I$.
	\end{itemize}
\end{Definition}

Actually, the details in this definition are adapted to our needs in this paper,
as is the following theorem.
The general and more technical versions
as well as the proof of can be found \cite{PogorzelskiS-16}.
See also \cite{OrnsteinW-87} for earlier results.

Roughly speaking, the following theorem provides,
in the setting of finitely generated amenable groups,
$\ve$-quasi covers for every set
with small enough boundary compared to its volume.
Additionally, the theorem also provides control on the fraction
covered by different tiles with uniform almost densities.
To quantify these densities, we use the standard notation
$\ceil b:=\inf\{z\in\Z\mid z\ge b\}=\inf\Z\isect\Ico b\infty$
for the smallest integer above $b\in\R$ and define,
for all $\ve>0$ and $i\in\N$,
\begin{equation}\label{eq:def:N,eta}
	N(\epsilon):=\biggceil{\frac{\ln(\epsilon)}{\ln(1-\epsilon)}}
	\qtextq{and}
	\eta_i(\epsilon):=\epsilon(1-\epsilon)^{N(\epsilon)-i}
	\text.
\end{equation}

\begin{Theorem}\label{thm:STP}
  Let~$G$ be a finitely generated amenable group,
	$(Q_n)$ a nested F{\o}lner sequence, and $\epsilon\in\Ioo0{1/10}$.
  Then there is a finite and strictly increasing selection of sets
  $K_i\in\{Q_n\mid n\in\N\}$, $i\in\{1,\dotsc,N(\epsilon)\}$, 
	with the following quasi tiling property.
  For each F{\o}lner sequence~$(\Lambda_j)$,
  there exists $j_0(\epsilon)\in\N$
  such that for all $j\geq j_0(\epsilon)$,
  the sets $K_i$, $i\in\{1,\dotsc,N(\epsilon)\}$,
  are an $\epsilon$-quasi tiling of~$\Lambda_j$.
	Moreover, for all $j\ge j_0(\ve)$ and all $i\in\{1,\dots,N(\epsilon)\}$,
	the proportion of~$\Lambda_j$ covered by the tile~$K_i$ satisfies
	\begin{equation}\label{eq:etdensity}
		\biggabs{\frac{\setsize{K_iT_i^j}}{\setsize{\Lambda_j}}-\eta_i(\ve)}
		\le\frac{\ve^2}{N(\ve)}
		\text,
	\end{equation}
	where~$T_i^j$ denotes the center set of the tile~$K_i$
	for the $\ve$-quasi cover of~$\Lambda_j$.
\end{Theorem}

To make full use of \cref{thm:STP},
we need some properties of the densities~$\eta_i(\ve)$.
\begin{Lemma}\label{la:etaN}
  For $N(\epsilon)$ and $\eta_i(\epsilon)$ as in \eqref{eq:def:N,eta},
	the following holds true.
  \begin{enumerate}[(a)]
  \item\label{etaNa} For each $\epsilon\in\Ioo01$ we have
    \begin{equation*}
      1-\epsilon
      \le\sum_{i=1}^{N(\epsilon)}\eta_i(\epsilon)
      =1-(1-\epsilon)^{N(\epsilon)}
      \le1\text.
    \end{equation*}
  \item\label{etaNb}
    For each $\epsilon\in\Ioo0{1/10}$ and $i\in\{1,\dotsc,N(\epsilon)\}$, we have
    \begin{equation*}
      \frac{\ve}{N(\ve)}
      \le\eta_i(\epsilon)
      \le\epsilon
      \text.
    \end{equation*}
  \item\label{etaNc} For a bounded sequence $(\alpha_i)_{i \in \N}$
    and $\epsilon\in\Ioo0{1/10}$ we have the inequality
    \begin{align*}
      \biggabs{\sum_{i=1}^{N(\epsilon)}\alpha_i\eta_i(\epsilon)}
      \leq A\sqrt{\epsilon} + A_\epsilon,
    \end{align*}
    where $A:=\sup\{\abs{\alpha_i}\mid i\in\N\}$ and $A_\epsilon:=\sup\{\abs{\alpha_i}\mid i\in\N,\ i\geq \epsilon^{-1/2}\}$.
    In particular,
    \begin{align*}
      \lim_{{\epsilon\dnto0}}
      \sum_{i=1}^{N(\epsilon)}\alpha_i\eta_i(\epsilon)
      \le\liminf_{i\to\infty}\abs{\alpha_i}\text.
    \end{align*}
\end{enumerate}
\end{Lemma}

\begin{proof}
  Part~\ref{etaNa} is an easy implication of the sum formula
	for the geometric series.
	We refer to \cite[Remark 4.3]{PogorzelskiS-16} for the details.
	\par
  Let us prove~\ref{etaNb}.
  By definition of $\eta_i(\epsilon)$ we have $\eta_i(\epsilon)\le\epsilon$.
	In order to see the other inequality, we note that
  \begin{align*}
    \eta_i(\epsilon)
    \ge\epsilon(1-\epsilon)^{N(\epsilon)-1}
    \ge\epsilon(1-\epsilon)^{\frac{\ln(\epsilon)}{\ln(1-\epsilon)}}
    =\epsilon^2\text.
  \end{align*}
  Thus, it is sufficient to show that $\epsilon\ge1/N(\epsilon)$.
  To this end, note that by definition of~$N(\epsilon)$
	the following holds true:
  \begin{align*}
    \epsilon N(\epsilon)\geq\frac{\epsilon\ln(\epsilon)}{\ln(1-\epsilon)}\text.
  \end{align*}
  Using the assumption $\epsilon\in\Ioo0{1/10}$,
  a short and elementary calculation shows that
  the last expression is bounded from below by~$1$.
	\par
  To verify part~\ref{etaNc}, set
	$N_\epsilon^*:=\floor{\epsilon^{-1/2}}
		:=\sup\Z\isect\Ioc{-\infty}{\ve^{-1/2}}$,
	and calculate as follows
  \begin{align*}
    \biggabs{\sum_{i=1}^{N(\epsilon)}\alpha_i\eta_i(\epsilon)}
    \le\biggabs{\sum_{i=1}^{N_\epsilon^*}\alpha_i\eta_i(\epsilon)}
			+\biggabs{\sum_{i=N_\epsilon^*+1}^{N(\epsilon)}\alpha_i\eta_i(\epsilon)}
    \le AN_\epsilon^*\epsilon+A_\epsilon
		\le A\sqrt\epsilon+A_\epsilon\text.
  \end{align*}
  Note that it is easy to show that for $0<\epsilon<1/10$
	we have $N(\epsilon)>N_\epsilon^*>0$, such that both sums are non-empty.
\end{proof}

Next, we derive a useful corollary of \cref{thm:STP}.

\begin{Corollary}\label{cor:density}
  Let a finitely generated group~$G$,
  a subset $\Lambda\in\cF$ and $\epsilon\in\Ioo0{1/2}$ be given.
  Assume furthermore that the sets $K_i\in\cF$, $i\in\{1,\dotsc,N(\epsilon)\}$,
  are an $\epsilon$-quasi tiling of~$\Lambda$
  with almost densities~$\eta_i(\epsilon)$ and center sets $T_i\in\cF$,
  $i\in\{1,\dotsc,N(\epsilon)\}$, satisfying~\eqref{eq:etdensity}.
  Then we have
		for each $i\in\{1,\dots,N(\epsilon)\}$,
    the inequality estimating the ``density'' of the tile~$K_i$:
    \begin{align*}
      \biggabs{\frac{\setsize{T_i}}{\setsize{\Lambda}}
        -\frac{\eta_i(\epsilon)}{\setsize{K_i}}}
      \le4\ve\frac{\eta_i(\epsilon)}{\setsize{K_i}}\text.
    \end{align*}
\end{Corollary}

\begin{proof}
  We fix $i\in\{1,\dotsc,N(\epsilon)\}$,
  employ $\epsilon$-disjointness and the density estimate~\eqref{eq:etdensity},
  and deduce
  \begin{align*}
    (1-\epsilon)\frac{\setsize{K_i}\setsize{T_i}}{\setsize{\Lambda}}
    \le\frac{\setsize{K_i T_i}}{\setsize{\Lambda}}
    \le\eta_i(\epsilon) +\frac{\epsilon^2}{N(\epsilon)}\text.
  \end{align*}
  Therefore, with part~\ref{etaNb} of \cref{la:etaN}, we get
  \begin{align*}
    \frac{\setsize{T_i}}{\setsize{\Lambda}}
			-\frac{\eta_i(\epsilon)}{\setsize{K_i}}
    \le\frac{\eta_i(\epsilon)+\frac{\epsilon^2}{N(\epsilon)}}
            {(1-\epsilon)\setsize{K_i}}
      -\frac{\eta_i(\epsilon)}{\setsize{K_i}}
    =\frac{\epsilon\eta_i(\epsilon)+\frac{\epsilon^2}{N(\epsilon)}}
          {(1-\epsilon)\setsize{K_i}}
    \le\frac{2\epsilon\eta_i(\epsilon)}
						{(1-\epsilon)\setsize{K_i}}
    \le\frac{4\epsilon\eta_i(\epsilon)}{\setsize{K_i}}\text.
  \end{align*}
  \Cref{eq:etdensity} gives also a bound for the other direction.
  To be precise, we use
  \begin{equation}\label{eq:Tbound}
    \eta_i(\epsilon)-\frac{\epsilon^2}{N(\epsilon)}
		\le\frac{\setsize{K_iT_i}}{\setsize{\Lambda}}
		\le\frac{\setsize{K_i}\setsize{T_i}}{\setsize{\Lambda}}
  \end{equation}
  and again part~\ref{etaNb} of \cref{la:etaN} to obtain
  \begin{align*}
    \frac{\setsize{T_i}}{\setsize{\Lambda}}
			-\frac{\eta_i(\epsilon)}{\setsize{K_i}}
		\geq\frac{\eta_i(\epsilon)-\frac{\epsilon^2}{N(\epsilon)}}{\setsize{K_i}}
			-\frac{\eta_i(\epsilon)}{\setsize{K_i}}
    =-\frac{\epsilon^2 }{N(\epsilon)\setsize{K_i}}
    \ge-\frac{\epsilon\eta_i(\epsilon) }{\setsize{K_i}}\text.
  \end{align*}
  This 
	implies the claimed bound.
\end{proof}

Finally, we provide a generalization of almost additivity
for sets which are not disjoint, but only $\epsilon$-disjoint.
The proof can be found in \cite[Lemma~5.23]{diss-fabian}.

\begin{Lemma}\label{la:quasialmostadd}
  Let~$G$ be a finitely generated group,
	$f$ an admissible field with boundary term~$b$,
	and $\epsilon\in\Ioo0{1/2}$.
  Then for any $\epsilon$-disjoint sets $Q_i$, $i\in\{1,\dots,k\}$,
  we have for each $\omega\in\Omega$:
  \begin{equation*}
    \biggnorm{f(Q,\omega)-\sum_{i=1}^kf(Q_i,\omega)}
    \le\epsilon(3\constf+9\constb)\setsize{Q}+3\sum_{i=1}^kb(Q_i)\text,
  \end{equation*}
  where $Q:=\bigcup_{i=1}^kQ_i$ and~$\constb$
  is the constant from~\ref{extensive} of Definition~\ref{def:admissible}.
\end{Lemma}

\section{Approximation via the empirical measure}\label{secboundary}

Given some F{\o}lner sequence~$(\Lambda_j)$ and an admissible field~$f$,
the aim of this section is the approximation of the expression
$\frac{f(\Lambda_j,\omega)}{\setsize{\Lambda_j}}$
using elements of a second F{\o}lner sequence~$(Q_n)$
and associated empirical measures, cf.~Lemma \ref{la:quasi-first-step}.
This second sequence needs to satisfy certain additional assumptions,
namely we need that $(Q_n)$ is nested and satisfies for the
correlation length $r\in\N_0$ from~\ref{M3} that the sequences
\begin{align}\label{eq:monotone0}
  \biggl(\frac{b(Q_n)}{\setsize{Q_n}}\biggr)\textq,
  \biggl(\frac{b(Q_n^r)}{\setsize{Q_n}}\biggr)\qtextq{and}
  \biggl(\frac{\setsize{\partial^r(Q_n)}}{\setsize{Q_n}}\biggr)
  \qtext{converge \emph{monotonically} to~$0$.}
\end{align}
That these sequences converge to zero is clear by the fact that~$(Q_n)$
is a F{\o}lner sequence and~$b$
a boundary term in the sense of \cref{def:admissible}.
In order to obtain the monotonicity in~\eqref{eq:monotone0},
we choose a subsequence of~$(Q_n)$.
These considerations show that each amenable group admits
a nested F{\o}lner sequence~$(Q_n)$ which satisfies~$\eqref{eq:monotone0}$.
These terms will be used in the error estimates
in the approximations throughout this text.
To abbreviate the notation, we define
\begin{equation}\label{eq:def-beta}
  \beta'_n:=\max\biggl\{\frac{b(Q_n)}{\setsize{Q_n}},
  \frac{b(Q_n^r)}{\setsize{Q_n}},
  \frac{\setsize{\partial^r(Q_n)}}{\setsize{Q_n}}\biggr\}
  \qtextq{and}
  \beta(\ve):=\beta'_1\sqrt\ve+\beta'_{\ceil{1/\sqrt\ve}}
\end{equation}
for $n\in\N$ and $\ve\in\Ioo0{1/10}$.
Note that $(\beta'_n)_n$ is a monotone sequence and converges to~$0$,
and that by \cref{la:etaN}\ref{etaNc}
\begin{equation}\label{eq:beta}
  \sum_{i=1}^{N(\ve)}\beta'_i\eta_i(\ve)
  \le\beta(\ve)
  \xto{\ve\dnto0}0\text.
\end{equation}
\begin{Remark}\label{rem:beta_le_ve}
  For the proof of \cref{thm:main},
  we additionally have to ensure $\beta'_n\le(2n)^{-1}$
  for all $n\in\N$ while taking the subsequences above.
  We will track the boundary terms throughout the paper and use~$\beta(\ve)$
  until the very end, where we simplify the result by applying
  \begin{equation*}
    \beta(\ve)
    =\beta'_1\sqrt\ve+\beta'_{\ceil{1/\sqrt\ve}}
    \le\frac12\sqrt\ve+\frac1{2\ceil{1/\sqrt\ve}}
    \le\sqrt\ve\text.
  \end{equation*}
  The cost of this additional condition on the boundary terms is that,
  via \cref{thm:STP}, $j_0(\ve)$ in \cref{thm:main} will potentially increase.
  But up to here, we deal only with the geometry of~$G$
  and still have that~$j_0(\ve)$ depends only on~$\ve$.

  Moreover, let us emphasize that when considering an admissible set~$\cU$
  the value~$\sqrt{\ve}$ gives a uniform bound on~$\beta(\epsilon)$ for all $f\in\cU$,
  since in this situation all $f\in\cU$ are almost additive with the same boundary term~$b$.
\end{Remark}

Define for an admissible field~$f$ and $\Lambda\in\cF$ the function
\begin{equation}\label{eq:def:fLambda}
  f_\Lambda\from\Omega_{\Lambda}\to\B\textq,\quad
  f_\Lambda(\omega):=f(\Lambda,\omega')
  \qquad\text{where }
  \omega'\in \Pi_{\Lambda}^{-1}(\{\omega \})\text.
\end{equation}
Note that by~\ref{local} of Definition~\ref{def:admissible}
we see that~$f_\Lambda$ is well-defined (and measurable).
In the situation where we insert elements of the F{\o}lner sequence~$(\Lambda_n)$
or $(\Lambda_n^r)$, for some $r\in\N_0$, we write
\begin{equation}\label{eq:def:f_n}
  f_n:=f_{\Lambda_n} \qquad\text{or}\qquad f_n^r:=f_{\Lambda_n^r}\text.
\end{equation}

For given $K,T\in\cF$ and $\omega\in\Omega$ we define the \emph{empirical measure} by
\begin{equation}
  L^{\omega}(K,T)\from\Borel(\Omega_{KT})\to\Icc01\textq,
  L^{\omega}(K,T)
    =\frac1{\setsize{T}}\sum_{t\in T}\unitmass{(\tau_t \omega)_{K}}\text.
\end{equation}

Given $\epsilon\in\Ioo0{1/10}$  and sequences $(\Lambda_j)$ and $(Q_j)$ as above,
we obtain by Theorem \ref{thm:STP} finite sets $K_i(\epsilon)$, $i=1,\dotsc,N(\epsilon)$
and (for~$j$ large enough) center sets $T_i^j(\ve)$ which form an $\epsilon$-quasi tiling of $\Lambda_j$.
In this setting, we use for given $\omega\in\Omega$, $\epsilon\in\Ioo0{1/10}$,
$r>0$, $i\in\{1,\dotsc,N(\epsilon)\}$ and $j\in \N$ large enough the notation
\begin{equation}\label{def:Lijromega}
  L^{\omega}_{i,j}(\epsilon) := L^{\omega}(K_{i}(\epsilon),T_i^j(\ve)) \quad\text{and}\quad
f_i(\epsilon):=f_{K_i(\epsilon)}
\end{equation}
as well as
\begin{equation}\label{def:Lijromegar}
  L^{r,\omega}_{i,j}(\epsilon) := L^{\omega}(K_{i}^r(\epsilon),T_i^j(\ve)) \quad\text{and}\quad
f_i^r(\epsilon):=f_{K_i^r(\epsilon)}.
\end{equation}
Here, the reader may recall that  $K_{i}^r(\epsilon)=K_i(\epsilon)\setminus \partial^r(K_i(\epsilon))$.

Moreover, we use for $\Lambda\in\cF$, a measurable $f\from\Omega_\Lambda\to\B$ and a measure~$\nu$ on $(\Omega_\Lambda,\cB(\Omega_\Lambda))$ the notation
\begin{equation*}
  \spr f\nu:=\int_{\Omega_{\Lambda}}f(\omega)\dd\nu(\omega)\text.
\end{equation*}

\begin{Lemma}\label{la:empmeasure2}
  Let~$f$ be an admissible field and let $K,T\in\cF$ and $\omega\in\Omega$.
  Then,
  \begin{equation*}
    \spr{f_K}{L^\omega(K,T)}
    =\frac1{\setsize T}\sum_{t\in T}f(Kt,\omega)
    \text.
  \end{equation*}
\end{Lemma}
\begin{proof}
  We calculate using linearity and~\ref{transitive}
  of Definition \ref{def:admissible}
  \begin{align*}
    \spr{f_K}{L^\omega(K,T)}&
    =\int_{\Omega_{K}} f_K(\omega')\d L^\omega(K,T)(\omega')
    =\frac1{\setsize{T}} \sum_{t\in T} \int_{\Omega_{K}} f_K(\omega')\d
      \unitmass{(\tau_t \omega)_{K}}(\omega')\\&
    =\frac1{\setsize{T}} \sum_{t\in T} f_K((\tau_t\omega)_{K})
    =\frac1{\setsize{T}} \sum_{t\in T} f(K t,\omega)\text.\qedhere
  \end{align*}
\end{proof}

We proceed with the first approximation Lemma.

\begin{Lemma}\label{la:quasi-first-step}
  Let $G$ be a finitely generated amenable group,
  let~$f$ be an admissible field and let~$(\Lambda_n)$
  and~$(Q_n)$ be F{\o}lner sequences, where~$(Q_n)$
  is additionally nested and satisfies \eqref{eq:monotone0}.
  Then, we have for all $\omega\in\Omega$ that
  \begin{equation}\label{eq:limit:empmeasure}
    \lim_{\ve\dnto0}\lim_{j\to\infty}
      \biggnorm{\frac{f(\Lambda_j,\omega)}{\setsize{\Lambda_j}}
        -\sum_{i=1}^{N(\epsilon)}\eta_i(\epsilon)
        \frac{\spr{f_{i}^r(\epsilon)}{L^{r,\omega}_{i,j}(\epsilon)}}
        {\setsize{ K_i(\epsilon)}}}
    =0\text.
  \end{equation}
  where $K_i(\epsilon)$, $i\in\{1,\dotsc,N(\epsilon)\}$
  are given by \cref{thm:STP}.
  Moreover, we have for arbitrary $\ve\in\Ioo0{1/10}$ and $j\ge j_0(\ve)$,
  with $j_0(\ve)$ from Theorem~\ref{thm:STP}, the inequality
  \begin{align*}
    \biggnorm{\frac{f(\Lambda_j,\omega)}{\setsize{\Lambda_j}}
      -\sum_{i=1}^{N(\epsilon)}\eta_i(\epsilon)
      \frac{\spr{f_{i}^r(\epsilon)}{L^{r,\omega}_{i,j}(\epsilon)}}
           {\setsize{K_i(\epsilon)}}
    }
    \le(9\constf+15\constb)\epsilon+12(2+\constf+\constb)\beta(\ve)\text.
  \end{align*}
\end{Lemma}
\begin{proof}
	Let $\epsilon\in\Ioo0{1/10}$ and $j\ge j_0(\epsilon)\in\N$ be given,
  where $j_0(\epsilon)$ is the constant given by Theorem~\ref{thm:STP}.
	We estimate using the triangle inequality:
	\begin{align}\label{eq:triangle1}
		\biggnorm{\frac{f(\Lambda_j,\omega)}{\setsize{\Lambda_j}}
      -\sum_{i=1}^{N(\epsilon)}\eta_i(\epsilon)
      \frac{\spr{f_{i}^r(\epsilon)}{L^{r,\omega}_{i,j}(\epsilon)}}
           {\setsize{ K_i(\epsilon)}}
    }
    \le a(\epsilon,j)+\sum_{i=1}^{N(\epsilon)}b_i(\epsilon,j)
      +\sum_{i=1}^{N(\epsilon)}c_i(\epsilon,j),
	\end{align}
	where
	\begin{align*}
    a(\epsilon,j)&
    :=\frac{1}{\setsize{\Lambda_j}}\biggnorm{f(\Lambda_j,\omega)
      -\sum_{i=1}^{N(\ve)}\sum_{t\in T_i^j(\ve)}f(K_i(\ve)t,\omega)}\text,\\
    b_i(\epsilon,j)&
    :=\biggnorm{\sum_{t\in T_i^j(\epsilon)}
      \frac{f(K_i(\epsilon)t,\omega)}{\setsize{\Lambda_j}}
      -\eta_i(\epsilon)\frac{\spr{f_i(\ve)}{L^\omega_{i,j}(\ve)}}
        {\setsize{K_i(\ve)}}}\text{, and}\\
    c_i(\epsilon,j)&
      :=\frac{\eta_i(\epsilon)}{\setsize{K_i(\epsilon)}}
      \bignorm{\spr{f_i(\ve)}{L^{\omega}_{i,j}(\ve)}
        -\spr{f_i^r(\ve)}{L^{r,\omega}_{i,j}(\ve)}}\text.
	\end{align*}
	Here, the expressions $L_{i,j}^\omega(\epsilon)$ and $f_i(\epsilon)$
  are given by \eqref{def:Lijromega}.
	Let us estimate the term $a(\epsilon,j)$.
  To this end, denote the part which is covered by translates of~$K_i(\ve)$,
  $i\in\{1,\dotsc,N(\epsilon)\}$ by
  \begin{equation*}
    R_i^j(\epsilon)
    :=\bigcup_{i=1}^{N(\epsilon)}K_i(\epsilon)T_i^j(\epsilon)
    \subseteq\Lambda_j\text.
  \end{equation*}
	Then we have, using the properties of the
  $\epsilon$-quasi tiling and part~\ref{etaNa} of Lemma~\ref{la:etaN},
  \begin{equation*}
    \setsize{R_i^j(\epsilon)}
    =\sum_{i=1}^{N(\epsilon)}\setsize{K_i(\epsilon)T_i^j(\epsilon)}
    \ge\setsize{\Lambda_j}\sum_{i=1}^{N(\epsilon)}
      \biggl(\eta_i(\epsilon)-\frac{\epsilon^2}{N(\epsilon)}\biggr)
	  \ge(1-2\epsilon)\setsize{\Lambda_j}\text,
  \end{equation*}
	which in turn gives
  $\setsize{\Lambda_j\setminus R_i^j(\ve)}\le2\ve\setsize{\Lambda_j}$.
	We use this and Lemma \ref{la:quasialmostadd} to calculate
	\begin{align*}
	\setsize{\Lambda_j}a(\epsilon,j) &\leq
	 (3\constf+9\constb)\epsilon\setsize{\Lambda_j}+3b(\Lambda_j\setminus R_i^j(\epsilon))  +\norm{f(\Lambda_j\setminus R_i^j(\epsilon))} +3\sum_{i=1}^{N(\epsilon)}\sum_{t\in T_i^j(\epsilon)}b(K_i(\epsilon)t )
	\\ &\leq
		(3\constf+9\constb)\epsilon\setsize{\Lambda_j} + (\constf+ 3 \constb)\setsize{\Lambda_j\setminus R_i^j(\epsilon)}
	+ 3\sum_{i=1}^{N(\epsilon)}\setsize{T_i^j(\epsilon)}b(K_i(\epsilon))
	\\ &\leq
		(5\constf+15\constb)\epsilon\setsize{\Lambda_j}
	+ 3\sum_{i=1}^{N(\epsilon)}\setsize{T_i^j(\epsilon)}b(K_i(\epsilon)).
	\end{align*}
	By $\epsilon$-disjointness and \eqref{eq:etdensity} we obtain
	\begin{align}\label{eq:est(iv)rechts}
		\frac12\setsize{K_i(\epsilon)}\setsize{T_i^j(\epsilon)}
		\le(1-\epsilon)\setsize{K_i(\epsilon)}\setsize{T_i^j(\epsilon)}
		\le\setsize{K_i(\epsilon)T_i^j(\epsilon)}
		\le\biggl(\eta_i(\epsilon)+\frac{\epsilon^2}{N(\epsilon)}\biggr)
			\setsize{\Lambda_j}\text,
	\end{align}
	which together with~\ref{etaNb} of Lemma~\ref{la:etaN} gives
	\begin{align*}
		\sum_{i=1}^{N(\epsilon)}\setsize{T_i^j(\epsilon)}b(K_i(\epsilon))&
		\le2\setsize{\Lambda_j}\sum_{i=1}^{N(\epsilon)}\biggl(\eta_i(\epsilon)
			+\frac{\epsilon^2}{N(\epsilon)}\biggr)
			\frac{b(K_i(\epsilon))}{\setsize{K_i(\epsilon)}}\\&
		\le4\setsize{\Lambda_j}\sum_{i=1}^{N(\epsilon)}
			\!\eta_i(\epsilon)\frac{b(K_i(\epsilon))}{\setsize{K_i(\epsilon)}}\text.
	\end{align*}
	This implies the following bound
	\begin{align}\label{eq:bound:a}
		a(\epsilon,j)
		\le(5\constf+15\constb)\ve+12\sum_{i=1}^{N(\epsilon)}
			\eta_i(\epsilon)\frac{b(K_i(\epsilon))}{\setsize{K_i(\epsilon)}}\text.
	\end{align}
	To estimate the second term in \eqref{eq:triangle1},
	we apply \cref{la:empmeasure2} to obtain
	\begin{align*}
	  \sum_{t\in T_i^j(\epsilon)}f(K_i(\epsilon)t,\omega)
		=\setsize{T_i^j(\epsilon)}\cdot\spr{f_i(\epsilon)}{L_{i,j}^\omega(\epsilon)}
		\text.
	\end{align*}
	Thus, by \cref{cor:density} and the fact
	$\norm{\spr{f_i(\epsilon)}{L_{i,j}^\omega(\epsilon)}}
		\le\constf\setsize{K_i(\epsilon)}$,
	we have for each $i\in\{1,\dots,N(\epsilon)\}$:
	\begin{align}
	  b_i(\epsilon,j)&
		=\biggnorm{\frac{\setsize{T_i^j(\epsilon)}
			\spr{f_i(\epsilon)}{L_{i,j}^\omega(\epsilon)}}{\setsize{\Lambda_j}}
			-\eta_i(\epsilon)\frac{\spr{f_{i}(\epsilon)}{L^{\omega}_{i,j}(\epsilon)}}
				{\setsize{K_i(\epsilon)}}}\nonumber\\&
		=\biggabs{\frac{\setsize{T_i^j(\epsilon)}}{\setsize{\Lambda_j}}
			-\frac{\eta_i(\epsilon)}{\setsize{K_i(\epsilon)}}}
			\norm{\spr{f_i(\epsilon)}{L_{i,j}^\omega(\epsilon)}}\nonumber\\&
		\le4\frac{\epsilon \eta_i(\epsilon)}{\setsize{K_i(\epsilon)}}
			\constf\setsize{K_i(\epsilon)}
		=4\constf\epsilon\eta_i(\epsilon)
		\label{eq:bound:b}\text.
	\end{align}

	Let us finally estimate the term $c_i(\epsilon,j)$.
	By Lemma \ref{la:empmeasure2} we have for each $i\in\{1,\dots,N(\epsilon)\}$
	\begin{align}&
		\bignorm{\spr{f_{i}(\epsilon)}{L^{\omega}_{i,j}(\epsilon)}
			-\spr{f_{i}^r(\epsilon)}{L^{r,\omega}_{i,j}(\epsilon)}}
		\le\frac1{\setsize{T_i^j(\epsilon)}}\sum_{t\in T_i^j(\epsilon)}
			\bignorm{f(K_i(\epsilon)t,\omega)-f(K_i^r(\epsilon)t,\omega)}
		\nonumber\\&
		\le\frac1{\setsize{T_i^j(\epsilon)}}\sum_{t\in T_i^j(\epsilon)}
			b(K_i^r(\epsilon))+b(\partial^r(K_i(\epsilon))\cap K_i(\epsilon))
			+\bignorm{f(\partial^r(K_i(\epsilon))t\cap K_i(\epsilon)t,\omega)}
		\nonumber\\&
		\le b(K_i^r(\epsilon))+(\constf+\constb)\setsize{\partial^r(K_i(\epsilon))}
		\text.\label{eq:bound:c}
	\end{align}
	Together with \eqref{eq:triangle1},
	the estimates for $a(\epsilon,j)$ in~\eqref{eq:bound:a},
	for $b_i(\epsilon,j)$ in~\eqref{eq:bound:b}
	and for $c_i(\epsilon,j)$ in~\eqref{eq:bound:c} yield
	\begin{align*}&
		\biggnorm{\frac{f(\Lambda_j,\omega)}{\setsize{\Lambda_j}}
			-\sum_{i=1}^{N(\epsilon)}\eta_i(\epsilon)
			\frac{\spr{f_{i}^r(\epsilon)}{L^{r,\omega}_{i,j}(\epsilon)}}
			{\setsize{ K_i(\epsilon)}}}\\&\quad
		\le(5\constf+15\constb)\ve+12\sum_{i=1}^{N(\epsilon)}
			\eta_i(\epsilon)\frac{b(K_i(\epsilon))}{\setsize{K_i(\epsilon)}}
		\\&\quad\quad
			+\sum_{i=1}^{N(\epsilon)}\eta_i(\epsilon)\biggl(4\constf\epsilon
			+\frac{b(K_i^r(\epsilon))+(\constf+\constb)\setsize{\partial^r(K_i(\epsilon))}}
						{\setsize{K_i(\epsilon)}}\biggr)\\&\quad
		\le(9\constf+15\constb)\ve+12\sum_{i=1}^{N(\epsilon)}
			\eta_i(\epsilon)\frac{b(K_i(\epsilon))+b(K_i^r(\epsilon))
				+(\constf+\constb)\setsize{\partial^r(K_i(\epsilon))}}
				{\setsize{K_i(\epsilon)}}\text.
	\end{align*}
	To verify~\eqref{eq:limit:empmeasure},
	recall that we assumed that~$(Q_n)$ satisfies \eqref{eq:monotone0}.
	By the choice of~$K_i(\epsilon)$ in \cref{thm:STP}, this gives
	\begin{align*}&
		\biggnorm{\frac{f(\Lambda_j,\omega)}{\setsize{\Lambda_j}}
			-\sum_{i=1}^{N(\epsilon)}\eta_i(\epsilon)
			\frac{\spr{f_{i}^r(\epsilon)}{L^{r,\omega}_{i,j}(\epsilon)}}
           {\setsize{ K_i(\epsilon)}}}\\&\quad
		\le(9\constf+15\constb)\ve+12\sum_{i=1}^{N(\ve)}\eta_i(\ve)
			\underbrace{\frac{b(Q_i)+b(Q_i^r)+(\constf+\constb)\setsize{\partial^r(Q_i)}}
											 {\setsize{Q_i}}}_{\le(2+\constf+\constb)\beta'_i}\\&\quad
		\le(9\constf+15\constb)\ve+12(2+\constf+\constb)\beta(\ve)
		\text.
	\end{align*}
  The last inequality follows from~\eqref{eq:beta}.
	As this bound holds for arbitrary $\epsilon\in\Ioo0{1/10}$ and
	$j\ge j_0(\epsilon)$, this particularly proves~\eqref{eq:limit:empmeasure}.
\end{proof}

\section{Approximation via Glivenko--Cantelli}\label{sec:GC}

In this section we aim to apply a multivariate Glivenko--Cantelli theorem
in order to approximate the empirical measure using the theoretical measure.
Recall that a Glivenko--Cantelli theorem compares the empirical measure
of a normalized sum of independent and identically
distributed random variables with their distribution.
At the end of this section, we will apply the following Glivenko--Cantelli
theorem which was proved in \cite{SchumacherSV-16}
based on results by DeHardt and Wright, see \cite{DeHardt1971,Wright1981}.
Monotone functions on~$\R^k$ were defined in~\ref{monotone}.
\begin{Theorem}\label{dewright}
  Let $(\Omega,\cA,\P)$ be a probability space and $X_t\from\Omega\to\R^k$,
  $t\in\N$, independent and identically distributed random variables
   such that the distribution $\mu:=\P(X\in\argmt)$
  is absolutely continuous with respect to a product measure
  $\Tensor_{\ell=1}^k\mu_\ell$ on~$\R^k$, where~$\mu_\ell$,
  $\ell\in\{1,\dotsc,k\}$, are $\sigma$-finite measures on~$\R$.
  For each $n\in\N$, we denote by
  $L_n^{(\omega)}:=\frac1n\sum_{t=1}^n\dirac{X_t}$
  the empirical distribution of $(X_t)_{t\in\{1,\dotsc,n\}}$.
  Further, fix $M\in\R$ and let
  $\cM:=\{g\from\R^k\to\R\mid\text{$g$ is monotone,
    and $\sup_{x\in\R^k}\abs{g(x)}\le M$}\}$.
  \par
  Then, for all $\kappa>0$, there are $a=a(\kappa,M)>0$ and $b=b(\kappa,M)>0$
  such that for all $n\in\N$, there exists an event $\Omega_{\kappa,n,M}\in\cA$
  with large probability
  $ 
    \P(\Omega_{\kappa,n,M})\ge1-b\exp(-an)\text,
  $ 
  such that for all $\omega\in\Omega_{\kappa,n,M}$, we have
  \begin{equation*}
    \sup_{g\in\cM}\abs{\spr g{L_n^{(\omega)}-\mu}}\le\kappa\text.
  \end{equation*}
  In particular, there exists a set $\Omega_0\in\cA$ with $\P(\Omega_0)=1$
  and $\sup_{g\in\cM}\abs{\spr g{L_n^{(\omega)}-\mu}}\xto{n\to\infty}0$
  for all $\omega\in\Omega_0$.
\end{Theorem}

In the present situation we encounter several challanges
when applying \cref{dewright}, caused by our tiling scheme.
\begin{itemize}
\item Each~$\Lambda_j$ is tiled using $N(\epsilon)$ different shapes.
  Thus, the corresponding random variables (for different shapes)
  are \emph{not identically distributed}.
\item In an $\epsilon$-quasi tiling, translates of the same shape~$K_i$
  are allowed to overlap.
  Thus, the corresponding random variables are
  \emph{not necessarily independent}.
\end{itemize}

The first point can be handled by applying Glivenko--Cantelli theory
for each shape~$K_i$ separately.
The second point is more challenging.
The core of the following approach is the ``generation of independence''
by resampling of the overlapping areas using conditional probabilities
and controlling errors introduced on the altered areas with their volume.
Let us explain this in detail.

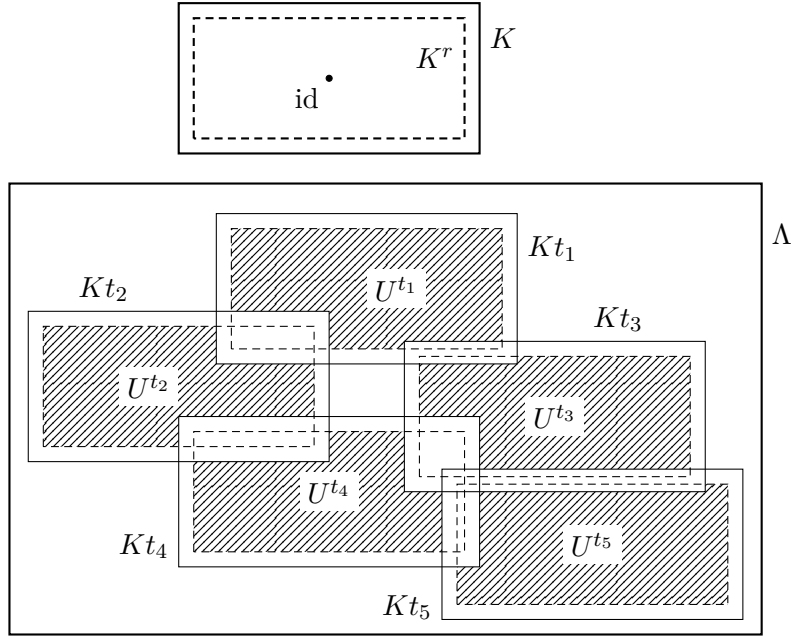
\begin{figure}
\begin{center}
\pgfooclass{foelner}{
  \attribute size x = 1;
  \attribute size y = 1;
  \attribute x = 0;
  \attribute y = 0;
  \attribute r = .1;

  \method foelner(#1,#2;#3,#4;#5) { 
    \pgfooset{size x}{#1}
    \pgfooset{size y}{#2}
    \pgfooset{x}{#3}
    \pgfooset{y}{#4}
    \pgfooset{r}{#5}
  }

  \method get size x(#1) { 
    \pgfooget{size x}{#1}
  }

  \method get size y(#1) { 
    \pgfooget{size y}{#1}
  }

  \method get x(#1) { 
    \pgfooget{x}{#1}
  }

  \method get y(#1) { 
    \pgfooget{y}{#1}
  }

  \method get r(#1) { 
    \pgfooget{r}{#1}
  }

  \method drawwithbourboundary(#1) { 
    \pgfooget{size x}{\sizex}
    \pgfooget{size y}{\sizey}
    \pgfooget{x}{\x}
    \pgfooget{y}{\y}
    \pgfooget{r}{\r}
    \coordinate (ul) at ($(\x,\y)-(\sizex,\sizey)$);
    \coordinate (or) at ($(\x,\y)+(\sizex,\sizey)$);
    \draw[#1] (ul) rectangle (or);
    \draw[#1,densely dashed] ($(ul)+\r*(1,1)$) rectangle ($(or)-\r*(1,1)$);
  }

  \method drawtile(#1) { 
    \pgfooget{size x}{\sizex}
    \pgfooget{size y}{\sizey}
    \pgfooget{x}{\x}
    \pgfooget{y}{\y}
    \coordinate (ul) at ($(\x,\y)-(\sizex,\sizey)$);
    \coordinate (or) at ($(\x,\y)+(\sizex,\sizey)$);
    \draw[#1] (ul) rectangle (or);
  }

  \method isect(#1) { 
    \foreach \o in {#1} {
      \pgfoothis.get size x(\sx)
      \pgfoothis.get size y(\sy)
      \pgfoothis.get x(\x)
      \pgfoothis.get y(\y)
      \pgfoothis.get r(\r)
      \o.get size x(\Sx)
      \o.get size y(\Sy)
      \o.get x(\X)
      \o.get y(\Y)
      \o.get r(\R)
      \pgfmathsetmacro{\ulx}{max(\x-\sx,\X-\Sx)}
      \pgfmathsetmacro{\uly}{max(\y-\sy,\Y-\Sy)}
      \pgfmathsetmacro{\orx}{min(\x+\sx,\X+\Sx)}
      \pgfmathsetmacro{\ory}{min(\y+\sy,\Y+\Sy)}
      \pgfmathtruncatemacro{\intersection}{(\ulx < \orx) && (\uly < \ory)}
      \ifthenelse{\equal{\intersection}{1}}
        {\draw[very thick, pattern=north east lines]
          (\ulx,\uly) rectangle (\orx,\ory);}
        {}
    }
  }

  \method mark coordinates(#1) { 
    \pgfoothis.get x(\x)
    \pgfoothis.get y(\y)
    \pgfoothis.get size x(\sx)
    \pgfoothis.get size y(\sy)
    \pgfoothis.get r(\r)
    \coordinate (#1) at (\x,\y);
    \coordinate (#1-ul) at (\x-\sx,\y+\sy);
    \coordinate (#1-ll) at (\x-\sx,\y-\sy);
    \coordinate (#1-ur) at (\x+\sx,\y+\sy);
    \coordinate (#1-lr) at (\x+\sx,\y-\sy);
    \pgfmathsetmacro{\sxr}{\sx-\r}
    \pgfmathsetmacro{\syr}{\sy-\r}
    \coordinate (#1-ul') at (\x-\sxr,\y+\syr);
    \coordinate (#1-ll') at (\x-\sxr,\y-\syr);
    \coordinate (#1-ur') at (\x+\sxr,\y+\syr);
    \coordinate (#1-lr') at (\x+\sxr,\y-\syr);
  }

  \method drawatcenter([#1]#2) { 
    \pgfoothis.get x(\x)
    \pgfoothis.get y(\y)
    \draw[#1] (\x,\y) #2;
  }

  \method markandnamecenter(#1,#2) { 
    \pgfoothis.get x(\x)
    \pgfoothis.get y(\y)
    \draw[fill] (\x,\y) circle (.4mm) node[#1] {#2};
  }

  \method boundary coordinate(#1,#2) { 
    \pgfoothis.get x(\x)
    \pgfoothis.get y(\y)
    \pgfoothis.get size x(\sx)
    \pgfoothis.get size y(\sy)
    \pgfmathsetmacro{\cx}{\x}
    \pgfmathsetmacro{\cy}{\y}
    \pgfmathtruncatemacro{\xpos}{cos(#1) >= 0
                              && \sx * sin(#1) >= -\sy * cos(#1)
                              && \sx * sin(#1) <= \sy * cos(#1)}
    \ifthenelse{\equal{\xpos}{1}} 
     {\pgfmathsetmacro{\cx}{\x + \sx}
      \pgfmathsetmacro{\cy}{\y + tan(#1) * \sx}}
     {
    \pgfmathtruncatemacro{\ypos}{(cos(#1) >= 0
                              && \sx * sin(#1) >= \sy * cos(#1))
                              || (cos(#1) <= 0
                              && \sx * sin(#1) >= -\sy * cos(#1))}
    \ifthenelse{\equal{\ypos}{1}} 
     {\pgfmathsetmacro{\cx}{\x + cot(#1) * \sy}
      \pgfmathsetmacro{\cy}{\y + \sy}}
     {
    \pgfmathtruncatemacro{\xneg}{cos(#1) <= 0
                              && \sx * sin(#1) <= -\sy * cos(#1)
                              && \sx * sin(#1) >= \sy * cos(#1)}
    \ifthenelse{\equal{\xneg}{1}} 
     {\pgfmathsetmacro{\cx}{\x - \sx}
      \pgfmathsetmacro{\cy}{\y - tan(#1) * \sx}}
     {
    \pgfmathtruncatemacro{\yneg}{(cos(#1) <= 0
                              && \sx * sin(#1) <= \sy * cos(#1))
                              || (cos(#1) >= 0
                              && \sx * sin(#1) <= -\sy * cos(#1))}
    \ifthenelse{\equal{\yneg}{1}} 
     {\pgfmathsetmacro{\cx}{\x - cot(#1) * \sy}
      \pgfmathsetmacro{\cy}{\y - \sy}}
     {\node {\Large This case cannot occur!};
    }}}}
    \coordinate (#2) at (\cx,\cy);
  }
}

\begin{tikzpicture}
  \pgfoonew\K=new foelner(2,1;0,3;0.2)
  \K.drawwithbourboundary(thick)
  \K.markandnamecenter(below left,$\id$)
  \K.mark coordinates(Kc)
  \K.boundary coordinate(15,Kboundary)
  \node[right] at (Kboundary) {$K$};
  \node[below left] at ($(Kboundary) - (.2,0)$) {$K^r$};
  \pgfoonew\U=new foelner(5,3;.75,-1.4;0)
  \U.drawtile(thick)
  \U.boundary coordinate(25,Uboundary)
  \node[right] at (Uboundary) {$\Lambda$};
  \pgfoonew\f=new foelner(2,1;.5,.2;0.2)     \f.mark coordinates(f)
  \pgfoonew\g=new foelner(2,1;-2,-1.1;0.2)   \g.mark coordinates(g)
  \pgfoonew\h=new foelner(2,1;3,-1.5;0.2)    \h.mark coordinates(h)
  \pgfoonew\i=new foelner(2,1;0,-2.5;0.2)    \i.mark coordinates(i)
  \pgfoonew\j=new foelner(2,1;3.5,-3.2;0.2)  \j.mark coordinates(j)
  \f.drawwithbourboundary()
  \f.boundary coordinate(15,fboundary)
  \node[right] at (fboundary) {$Kt_1$};
  \g.drawwithbourboundary()
  \g.boundary coordinate(135,gboundary)
  \node[above] at (gboundary) {$Kt_2$};
  \h.drawwithbourboundary()
  \h.boundary coordinate(50,hboundary)
  \node[above] at (hboundary) {$Kt_3$};
  \i.drawwithbourboundary()
  \i.boundary coordinate(200,iboundary)
  \node[left] at (iboundary) {$Kt_4$};
  \j.drawwithbourboundary()
  \j.boundary coordinate(195,jboundary)
  \node[below left] at (jboundary) {$Kt_5$};
  \fill[pattern=north east lines]
    (f-ul') -- (f-ul' |- g-ur) -- (g-ur) -- (g-ur |- f-lr')
    -- (f-lr' -| h-ul) -- (h-ul) -- (h-ul -| f-ur') -- (f-ur') -- cycle;
  \fill[pattern=north east lines]
    (g-ul') -- (g-ll') -- (g-ll' -| i-ul) -- (i-ul) -- (i-ul -| g-lr')
    -- (g-lr' |- f-ll) -- (f-ll) -- (f-ll |- g-ur') -- cycle;
  \fill[pattern=north east lines]
    (h-ur') -- (h-ur' |- j-ul) -- (j-ul -| i-ur) -- (i-ur) -- (i-ur -| h-ul')
    -- (h-ul' |- f-lr) -- (f-lr) -- (f-lr |- h-ur') -- cycle;
  \fill[pattern=north east lines]
    (i-ll') -- (i-ll' -| j-ul) -- (j-ul |- h-ll) -- (h-ll) -- (h-ll |- i-ur')
    -- (i-ur' -| g-lr) -- (g-lr) -- (g-lr -| i-ll') -- cycle;
  \fill[pattern=north east lines]
    (j-ll') -- (j-lr') -- (j-ur') -- (j-ur' -| h-lr) -- (h-lr)
    -- (h-lr -| i-lr) -- (i-lr) -- (i-lr -| j-ll') -- cycle;
  \node[right,fill=white,inner sep=1mm] at (f) {$U^{t_1}$};
  \node[left,fill=white,inner sep=1mm] at (g) {$U^{t_2}$};
  \node[fill=white,inner sep=1mm] at (h) {$U^{t_3}$};
  \node[fill=white,inner sep=1mm] at (i) {$U^{t_4}$};
  \node[fill=white,inner sep=1mm] at (j) {$U^{t_5}$};
\end{tikzpicture}
\end{center}
\caption{$\epsilon$-covering and independence structure:
  The set~$\Lambda=\Lambda_j\subseteq G$
  is $\epsilon$-quasi covered by copies of~$K=K_i$
  with centers in~$T=T_i^j(\epsilon)=\{t_1,\dotsc,t_5\}$.
  The sets $U^t=U^{i,j,t}$, 
  $t\in T$, here marked by diagonal stripes, have at least distance~$r$
  and satisfy $\setsize{U^t}\ge(1-\epsilon)\setsize K$.}
\label{fig:vecovering}
\end{figure}

Fix $\ve>0$, $i\in\{1,\dotsc,N(\ve)\}$ and $j\in\N$, $j\ge j_0(\ve)$,
cf.\ \cref{thm:STP},
and consider \Cref{fig:vecovering}, which sketches a tile $K=K_i$,
a finite set $\Lambda=\Lambda_j$, and the translates~$Kt$, $t\in T:=T_i^j(\ve)$,
of $K=K_i$ from an $\epsilon$-quasi tiling.
The sets
\begin{equation}\label{eq:Uijt}
  U^{i,j,t}
  :=(K_i^rt)\setminus\bigl(K_i(T_i^j(\ve)\setminus\{t\})\bigr)
  \subseteq G\text,\quad t\in T\text,
\end{equation}
are marked with stripes.
Their distance is at least
\begin{equation}\label{eq:distUU'}
  d(U^{i,j,t},U^{i,j,t'})
  \ge d(K_i^rt,G\setminus K_it)
  >r
  \text,
  \quad t\ne t'
  \text,
\end{equation}
so the colors there are $\P$-independent from each other.
Unfortunately, if we take only the values on~$U^{i,j,t}$, $t\in T$,
we will end up with an independent, but not identically distributed sample.
We therefore resample independent colors in $K^r\setminus U^{i,j,t}$.
Fortunately, the sets~$U^{i,j,t}$
are large enough to compensate this small random perturbation.
The following lemma specifies the resampling procedure.

\begin{Lemma}\label{lemma:independence}
  Let $\ve>0$ and
  $I:=\Union_{i=1}^{N(\ve)}\Union_{j=j_0(\ve)}^\infty\{(i,j)\}\times T_i^j(\ve)$.
  There exists a probability space
  $(\superOmega,\Borel(\superOmega),\superP)$
  and random variables $X,X^{i,j,t}\from\superOmega\to\Omega$,
  $(i,j,t)\in I$, such that for all $(i,j,t)\in I$,
  \begin{enumerate}[(i),nosep]
  \item\label{ind:distr} $X$ and $X^{i,j,t}$ have distribution~$\P$,
  \item\label{ind:keep} $X$ and~$X^{i,j,t}$ agree on $U^{i,j,t}$
    $\superP$-almost surely, and
  \item\label{ind:ind} the random variables in the set
    $\{X^{i,j,t'}\}_{t'\in T_i^j(\ve)}$ are $\superP$-independent.
  \end{enumerate}
\end{Lemma}
\begin{proof}
  \Cref{thm:resampling}
  solves the problem of resampling in an abstract setting.
  We apply the result here as follows.
  Since we use the canonical probability space in our construction,
  we apply \cref{thm:resampling} with
  $(S,\cS):=(\Omega,\cA)$, $X:=\id_{\Omega}$,
  $I:=\Union_{i=1}^{N(\ve)}\Union_{j=j_0(\ve)}^\infty\{(i,j)\}\times T_i^j(\ve)$,
  and $\cY_{j'}:=\sigma(\Pi_{U^{j'}})$, $j'\in I$.
  \Cref{thm:resampling} provides the following quantities,
  which we here want to use as
  $(\superOmega,\superA,\superP):=(\superOmega,\superA,\superP)$,
  $X:=\superX$, and
  $X^{i,j,t}:=X_{j'}$ for all $j'=(i,j,t)\in I$.
  The properties~\ref{ind:distr} and~\ref{ind:keep}
  follow directly from \cref{thm:resampling}\ref{resampling:distribution},%
  \ref{resampling:equality}.
  With~\eqref{eq:distUU'}, \cref{thm:resampling}\ref{resampling:independ}
  implies~\ref{ind:ind}.
\end{proof}

Next, we control the error we introduce by using our independent
samples instead of the dependent ones.
\begin{Lemma}\label{lemma:insertindependence}
  Fix $\ve>0$, an admissible~$f$ and $U\subseteq K\in\cF$.
  For $\omega,\tilde\omega\in\Omega$ with $\omega_U=\tilde\omega_U$,
  we have
  \begin{equation*}
    \norm{f(\omega,K)-f(\tilde\omega,K)}
    \le 2b(K)+2(2\constb+\constf)\setsize{K\setminus U}\text.
  \end{equation*}
  In particular, in the notation from
  \eqref{eq:def:fLambda} -- \eqref{def:Lijromegar}
  and with the corresponding empirical measure
  \begin{equation*}
    \superL_{i,j}^{r,\superomega}(\ve)
    :=\frac1{\setsize{T_i^j(\ve)}}
      \sum_{t\in T_i^j(\ve)}\dirac{(\tau_tX^{i,j,t}(\superomega))_{K_i(\ve)}}
    \qquad(\superomega\in\superOmega)\text,
  \end{equation*}
  we have for $\superP$-almost all $\superomega\in\superOmega$ that
  \begin{equation*}
    \norm{\spr{f_i^r(\ve)}
      {L_{i,j}^{r,X(\superomega)}(\ve)-\superL_{i,j}^{r,\superomega}(\ve)}}
    \le2b(K_i^r(\ve))+2(2\constb+\constf)\ve\setsize{K_i^r(\ve)}\text.
  \end{equation*}
\end{Lemma}
\begin{proof}
  The values of~$\omega$ on~$U$ determine $f(\omega,K)$ up to
  \begin{align*}
    \norm{f(\omega,K)-f(\omega,U)}&
    \le\norm{f(\omega,K)-f(\omega,U)-f(\omega,K\setminus U)}
      +\norm{f(\omega,K\setminus U)}\\&
    \le b(U)+b(K\setminus U)+\norm{f(\omega,K\setminus U)}
    \le b(U)+(\constb+\constf)\setsize{K\setminus U}\text.
  \end{align*}
  With the fourth point in~\ref{extensive}, we can continue this estimate with
  \begin{align*}
    b(U)&
    \le b(K\setminus(K\setminus U))
    \le b(K)+b(K\setminus U)
    \le b(K)+\constb\setsize{K\setminus U}\text.
  \end{align*}
  We now employ the triangle inequality to show the first claim:
  For $\omega,\tilde\omega\in\Omega$ with $\omega_U=\tilde\omega_U$, we have
  \begin{align*}
    \norm{f(\omega,K)-f(\tilde\omega,K)}&
    \le\norm{f(\omega,K)-f(\omega,U)}
      +\norm{f(\tilde\omega,U)-f(\tilde\omega,K)}\\&
    \le 2\bigl(b(K)+(2\constb+\constf)\setsize{K\setminus U}\bigr)\text.
  \end{align*}
  This calculation allows us to change~$\omega$ on $K\setminus U$
  to the independent values provided by \cref{lemma:independence}.
  To implement this, observe that for $\superP$-almost all
  $\superomega\in\superOmega$ and all $i\in\{1,\dotsc,N(\ve)\}$,
  $j\in\N$, $j\ge j_0(\ve)$ and $t\in T_i^j(\ve)$, the set
  $U^{i,j,t}$ from \eqref{eq:Uijt}
  exhausts $K_i^r(\ve)t$ up to a fraction of~$\ve$:
  $\setsize{K_i^r(\ve)t\setminus U^{i,j,t}}\le\ve\setsize{K_i^r(\ve)}$.
  By construction, on~$U^{i,j,t}$, the colors are preserved:
  $U^{i,j,t}\subseteq\{g\in K_i^r(\ve)t\mid X_g(\superomega)=X_g^{i,j,t}(\superomega)\}$.
  Together with \cref{la:empmeasure2} and the triangle inequality,
  this immediately implies for
  $\superP$-almost all $\superomega\in\superOmega$ that
  \begin{align*}
    \norm{\spr{f_i^r(\ve)}{L_{i,j}^{r,X(\superomega)}(\ve)-\superL_{i,j}^{r,\tilde\omega}(\ve)}}&
    \le\frac1{\setsize{T_i^j(\ve)}}\sum_{t\in T_i^j(\ve)}
      \norm{f(K_i^r(\ve)t,\omega)-f(K_i^r(\ve)t,X^{i,j,t}(\superomega))}\\&
    \le2b(K_i^r(\ve))+2(2\constb+\constf)\ve\setsize{K_i^r(\ve)}\text.\qedhere
  \end{align*}
\end{proof}

The empirical measure $L_{i,j}^{r,X(\superomega)}$
formed by independent samples should converge to
\begin{equation*}
  \P_i^r(\epsilon):=\P_{K_i^r(\epsilon)}\text.
\end{equation*}
The following result makes this notion precise.
It is the main result of this section.
\begin{Proposition}\label{la:gc-main-lemma}
  Let~$G$ be a finitely generated amenable group,
  let $\cA\in\Borel(\R)$ and $(\Omega:=\cA^G,\cB(\Omega),\P)$
  a probability space such that~$\P$ satisfies \labelcref{M1,M2,M3}.
  Moreover, let~$(\Lambda_n)$ and~$(Q_n)$ be F{\o}lner sequences,
  where~$(Q_n)$ is nested and satisfies \eqref{eq:monotone0}.
  For given $\ve\in\Ioo0{1/10}$, let~$K_i(\epsilon)$,
  $i\in\{1,\dots,N(\epsilon)\}$, and $j_0(\ve)$ be given by \cref{thm:STP}.
  Furthermore, let~$\cU$ be an admissible set of admissible fields.
  \par
  Then, for all $\kappa>0$, there exist
  $a(\ve,\kappa,\constf[\cU]),b(\ve,\kappa,\constf[\cU])>0$
  such that for all $j\ge j_0(\ve)$,
  there is an event $\Omega_{j,\ve,\kappa,\constf[\cU]}\in\Borel(\Omega)$
  with large probability
  \begin{equation*}
    \P(\Omega_{j,\ve,\kappa,\constf[\cU]})
    \ge1-b(\ve,\kappa,\constf[\cU])
      \exp(-a(\ve,\kappa,\constf[\cU])\setsize{\Lambda_j})
  \end{equation*}
  and the property that for all $\omega\in\Omega_{j,\ve,\kappa,\constf[\cU]}$
  and $f\in\cU$, it holds true that
  \begin{equation*}
    \biggnorm{\sum_{i=1}^{N(\epsilon)}\eta_i(\epsilon)
      \frac{\spr{f_{i}^r(\epsilon)}{L^{r,\omega}_{i,j}(\epsilon)}}
           {\setsize{ K_i(\epsilon)}}
      -\sum_{i=1}^{N(\epsilon)}\eta_i(\epsilon)
      \frac{\spr{f_{i}^r(\epsilon)}{\P_i^r(\epsilon)}}
           {\setsize{ K_i(\epsilon)}}}
    \le2\beta(\ve)+2(2\constb+\constf)\ve+\kappa\text.
  \end{equation*}
  In particular, there is an event $\tilde\Omega\in\cB(\Omega)$
  with $\P(\tilde\Omega)=1$ such that for all $\omega\in\tilde\Omega$,
  we have
  \begin{equation*}
    \lim_{\ve\dnto0}\sup_{f\in\cU}
    \biggnorm{\sum_{i=1}^{N(\epsilon)}\eta_i(\epsilon)
      \frac{\spr{f_{i}^r(\epsilon)}{L^{r,\omega}_{i,j}(\epsilon)}}
           {\setsize{ K_i(\epsilon)}}
      -\sum_{i=1}^{N(\epsilon)}\eta_i(\epsilon)
      \frac{\spr{f_{i}^r(\epsilon)}{\P_i^r(\epsilon)}}
           {\setsize{ K_i(\epsilon)}}}
    =0\text.
  \end{equation*}
\end{Proposition}

\begin{proof}
  Fix $f\in\cU$.
  For $\ve\in\Ioo0{1/10}$, $j\in\N$ and $\omega\in\Omega$,
  two applications of the triangle inequality give
  \begin{align}\label{Delta1}
  \Delta_f(\ve,\omega)&:=
  \biggnorm{\sum_{i=1}^{N(\epsilon)}\eta_i(\epsilon)
    \frac{\spr{f_{i}^r(\epsilon)}{L^{r,\omega}_{i,j}(\epsilon)}}
         {\setsize{ K_i(\epsilon)}}
    -\sum_{i=1}^{N(\epsilon)}\eta_i(\epsilon)
      \frac{\spr{f_{i}^r(\epsilon)}{\P_i^r(\epsilon)}}
           {\setsize{K_i(\epsilon)}}}\notag\\&
  \le\sum_{i=1}^{N(\epsilon)}\frac{\eta_i(\epsilon)}{\setsize{K_i(\epsilon)}}
    \bignorm{\spr{f_{i}^r(\epsilon)}{L^{r,\omega}_{i,j}(\epsilon)
      -\P_i^r(\epsilon)}}\notag\\&
  \le\inf_{\superomega\in X^{-1}(\{\omega\})}\biggl(
     \sum_{i=1}^{N(\epsilon)}\eta_i(\epsilon)\gamma_1(i,j,\epsilon,\superomega)
    +\sum_{i=1}^{N(\epsilon)}\eta_i(\epsilon)\gamma_2(i,j,\epsilon,\superomega)
    \biggr)\text,
  \end{align}
  where $\superomega\in\superOmega$ extends~$\omega$,
  i.\,e.\ $X(\superomega)=\omega$ in the notation of \cref{lemma:independence},
  and
  \begin{align*}
    \gamma_1(i,j,\epsilon,\superomega)&
    :=\frac{\bignorm{\spr{f_{i}^r(\epsilon)}{L^{r,\omega}_{i,j}(\epsilon)
      -\superL_{i,j}^{r,\superomega}(\ve)}}}{\setsize{K_i(\ve)}}
    \qtext{and}\\
    \gamma_2(i,j,\epsilon,\superomega)&
    :=\frac{\bignorm{\spr{f_i^r(\ve)}{\superL_{i,j}^{r,\superomega}(\ve)
      -\P_i^r(\epsilon)}}}{\setsize{K_i(\ve)}}
    \text.
  \end{align*}

  By \cref{lemma:insertindependence} and assumption \eqref{eq:monotone0},
  we see that for all $\superomega\in\superOmega$ with $X(\superomega)=\omega$
  \begin{equation*}
    \gamma_1(i,j,\ve,\superomega)
    \le\frac{2b(K_i^r(\ve))}{\setsize{K_i^r(\ve)}}+2(2\constb+\constf)\ve
    \le\frac{2b(Q_i)}{\setsize{Q_i}}+2(2\constb+\constf)\ve\text.
  \end{equation*}
  With \cref{la:etaN}\ref{etaNa} and \eqref{eq:beta},
  we yield the deterministic upper bound
  \begin{equation*}
    \sum_{i=1}^{N(\ve)}\eta_i(\ve)\gamma_1(i,j,\ve,\superomega)
    \le2\beta(\ve)+2(2\constb+\constf)\ve
  \end{equation*}
  for all $\superomega\in X^{-1}(\omega)\subseteq\superOmega$.
  By now, our overall inequality~\eqref{Delta1} reads
  \begin{equation}\label{Delta2}
    \Delta_f(\ve,\omega)
    \le2\beta(\ve)+2(2\constb+\constf)\ve
    +\inf_{\superomega\in X^{-1}(\{\omega\})}
      \sum_{i=1}^{N(\epsilon)}\eta_i(\epsilon)\gamma_2(i,j,\ve,\superomega)
    \text.
  \end{equation}

  To deal with~$\gamma_2$,
  recall that the norm on the Banach space~$\B$
  our admissible fields map into is the $\sup$-norm.
  We translate the $\sup$-norm into the Glivenko--Cantelli setting as follows.
  Let
  \begin{equation*}
    \cM_f:=\{g_{i,E}^r\from\R^{\setsize{K_i^r(\ve)}}\to\R,
    g_{i,E}^r(\omega):=f_i^r(\omega)(E)/\setsize{K_i(\ve)}\mid E\in\R\}\text.
  \end{equation*}
  Therefore, we can write
  \begin{equation*}
    \gamma_2(i,j,\ve,\superomega)
    =\sup_{g\in\cM_f}\abs{
      \spr g{\superL_{i,j}^{r,\superomega}(\ve)-\P_i^r(\ve)}}
    \le\sup_{f\in\cU}\sup_{g\in\cM_f}\abs{
      \spr g{\superL_{i,j}^{r,\superomega}(\ve)-\P_i^r(\ve)}}
    \text.
  \end{equation*}
  From \eqref{eq:bound-K_f} we see that the fields in
  $\cM_\cU:=\Union_{f\in\cU}\cM_f$ are bounded by~$\constf[\cU]$.
  As assumed in \ref{monotone}, the fields in~$\cM_\cU$ are also monotone.
  By \cref{lemma:independence}\ref{ind:ind}, the samples are independent, too.
  This is crucial in order to invoke \cref{dewright}.
  We thus obtain that, for each $\kappa>0$, $\ve\in\Ioo0{1/10}$,
  $i\in\{1,\dotsc,N(\ve)\}$ and $j\in\N$, $j\ge j_0(\ve)$,
  there are $a_i\equiv a(i,\ve,\kappa,\constf[\cU])>0$,
  $b_i\equiv b(i,\ve,\kappa,\constf[\cU])>0$ and $\superOmega_{i,j}\equiv
    \superOmega_{i,j,\ve,\kappa,\constf[\cU]}\in\Borel(\superOmega)$ such that
  \begin{equation*}
    \superP(\superOmega_{i,j})\ge1-b_i\exp(-a_i\setsize{T_i^j(\ve)})
    \qtextq{and}
    \sup_{\superomega\in\superOmega_{i,j}}
      \gamma_2(i,j,\ve,\superomega)\le\kappa
    \text.
  \end{equation*}
  We need this estimate for all~$i\in\{1,\dotsc,N(\ve)\}$
  simultaneously and consider
  \begin{equation*}
    \superOmega_j\equiv\superOmega_{j,\ve,\kappa,\constf[\cU]}
    :=\Isect_{i=1}^{N(\ve)}\superOmega_{i,j}\text.
  \end{equation*}
  To estimate the probability of~$\superOmega_j$ is the next step.
  From \eqref{eq:Tbound} and \cref{la:etaN}\ref{etaNb}, we note that
  \begin{equation*}
    \setsize{T_i^j(\ve)}
    \ge\Bigl(\eta_i(\ve)-\frac{\ve^2}{N(\ve)}\Bigr)
      \frac{\setsize{\Lambda_j}}{\setsize{K_i(\ve)}}
    \ge\frac{(1-\ve)\ve}{N(\ve)\setsize{K_i(\ve)}}\setsize{\Lambda_j}
    \text.
  \end{equation*}
  With the definition
  \begin{equation*}
    a\equiv a_{\ve,\kappa,\constf[\cU]}
    :=\frac{(1-\ve)\ve}{N(\ve)}
      \min_{i\in\{1,\dotsc,N(\ve)\}}\frac{a_i}{\setsize{K_i(\ve)}}
    \qtextq{and}
    b\equiv b_{\ve,\kappa,\constf[\cU]}
    :=2\sum_{i=1}^{N(\ve)\}}b_i
    \text,
  \end{equation*}
  we get $\superP(\superOmega_{i,j})\ge1-b_i\exp(-a\setsize{\Lambda_j})$
  and
  \begin{equation*}
    \superP(\superOmega_j)
    =1-\superP\Bigl(\Union_{i=1}^{N(\ve)}
      \superOmega\setminus\superOmega_{i,j}\Bigr)
    \ge1-\sum_{i=1}^{N(\ve)}\superP(\superOmega\setminus\superOmega_{i,j})
    \ge1-\frac{b\exp(-a\setsize{\Lambda_j})}2\text.
  \end{equation*}

  Next, we should transition from
  $(\superOmega,\Borel(\superOmega),\superP)$ to $(\Omega,\Borel(\Omega),\P)$.
  The set~$X(\superOmega_j)\subseteq\Omega$ seems to be a good candidate,
  because for all $\omega\in X(\superOmega_j)$, there exists
  $\superomega\in X^{-1}(\{\omega\})
    \isect\Isect_{i=1}^{N(\ve)}\superOmega_{i,j}$,
  and thus we can estimate
  \begin{align*}
    \inf_{\superomega\in X^{-1}(\{\omega\})}
    \sum_{i=1}^{N(\epsilon)}\eta_i(\epsilon)\gamma_2(i,j,\ve,\superomega)
    \le\sum_{i=1}^{N(\epsilon)}\eta_i(\epsilon)\kappa
    \le\kappa\text.
  \end{align*}
  Together with~\eqref{Delta2},
  this inequality shows the claimed bound on~$\Delta_f(\ve,\omega)$
  for all $\omega\in X(\superOmega_j)$.

  Unfortunately, the image of a measurable set under a measurable map
  is not necessarily measurable, but only analytic,
  see \cite[Theorem~10.23]{AliprantisBorder2006}.
  At least the outer measure of our candidate is bounded from below by
  \begin{align*}
    \P^*(X(\superOmega_j))&
    :=\inf_{B\in\Borel(\Omega),
        X(\superOmega_j)\subseteq B}\P(B)
    =\inf_{B\in\Borel(\Omega),
      X(\superOmega_j)\subseteq B}\superP(X\in B)\\&
    \ge\inf_{B\in\Borel(\Omega),X(\superOmega_j)\subseteq B}
      \superP(\superOmega_j)
    =\superP(\superOmega_j)
    \ge1-b\exp(-a\setsize{\Lambda_j})/2\text.
  \end{align*}
  From \cite[Lemma~10.36]{AliprantisBorder2006},
  we learn that~$\P^*$ is a nice capacity, and the Choquet Capacity Theorem
  \cite[Theorem~10.39]{AliprantisBorder2006}
  states for the analytic set~$X(\superOmega_j)$ that
  \begin{equation*}
    \P^*(X(\superOmega_j))
    =\sup_{K\subseteq X(\superOmega_j)\text{ compact}}\P(K)
    \text.
  \end{equation*}
  Thus, there exists a compact subset
  $\Omega_{j,\ve,\kappa,\constf[\cU]}\subseteq
    X(\superOmega_j)$
  with probability at least $1-b\exp(-a\setsize{\Lambda_j})$.

  We finish the proof with a standard Borel--Cantelli argument to show that
  $\tilde\Omega$ exists as claimed.
  For all $\kappa>0$, the events
  \begin{equation*}
    A_\kappa
    :=\Union_{n=j_0(\ve)}^\infty\Isect_{j=n}^\infty
      \Omega_{j,\ve,\kappa,\constf[\cU]}
  \end{equation*}
  have probability~$1$, since
  \begin{equation*}
    \sum_{j=j_0(\ve)}^\infty\P(\Omega\setminus\Omega_{j,\ve,\kappa,\constf[\cU]})
    \le\sum_{j=j_0(\ve)}^\infty b\exp(-a\setsize{\Lambda_j})
    \le b\sum_{j=j_0(\ve)}^\infty\exp(-a)^j
    <\infty\text.
  \end{equation*}
  Note that by \eqref{Delta2}, $\beta(\ve)\to0$, and by construction of~$A_k$,
  for all $\omega\in A_\kappa$, we have
  \begin{equation*}
    \lim_{\ve\dnto0}\sup_{f\in\cU}\Delta_f(\ve,\omega)\le\kappa\text.
  \end{equation*}
  Thus, the event $\tilde\Omega:=\Isect_{k\in\N}A_{1/k}$
  has full probability $\P(\tilde\Omega)=1$,
  and for all $\omega\in\tilde\Omega$, we have
  $\lim_{\ve\dnto0}\sup_{f\in\cU}\Delta_f(\ve,\omega)=0$.
\end{proof}

\section{Almost additivity and Cauchy sequences}\label{sec:cauchy}

The following calculations are devoted to a Cauchy sequence argument
to obtain the desired limit function~$f^*$.
\begin{Lemma}\label{la:cauchy}
  Let $G$ be a finitely generated amenable group,
  let $\cA\in\Borel(\R)$ and $(\Omega=\cA^G,\cB(\Omega),\P)$
  a probability space such that $\P$ satisfies \labelcref{M1,M2,M3}.
  Moreover, let $f$ be an admissible field
  and $(Q_n)$ a nested F{\o}lner sequence satisfying
  \eqref{eq:monotone0}.
  Then, there exists $f^*\in\B$ with
  \begin{align*}&
    \lim_{\ve\dnto0}\biggnorm{
    \sum_{i=1}^{N(\ve)}\eta_i(\ve)
    \frac{\spr{f_{i}^r(\ve)}{\P_i^r(\ve)}}
         {\setsize{ K_i(\ve)}}-f^*}=0\text,
  \end{align*}
  where for $k\in\N$ and $\epsilon\in\Ioo{1/(k+1)}{1/k}$
  the sets $K_i(\epsilon)$, $i\in\{1,\dots,N(\epsilon)\}$
  are extracted from the sequence $(Q_{n+k})_n$ via Theorem \ref{thm:STP}.
  The approximation error is bounded by
  \begin{equation*}
    \biggnorm{\sum_{j=1}^{N(\ve)}\eta_j(\ve)
      \frac{\spr{f_{j}^r(\ve)}{\P_j^r(\ve)}}{\setsize{ K_j(\ve)}}
      -f^*}
    \le(9\constf+11\constb)\ve+5(4+\constf+\constb)\beta(\ve)
    \text.
  \end{equation*}
\end{Lemma}

\begin{proof}
  In order to prove the existence of $f^*$,
  we study for $\epsilon,\delta\in\Ioo0{1/10}$ the difference
\begin{align*}
 \mathcal D(\epsilon,\delta):=\biggnorm{\sum_{j=1}^{N(\epsilon)}\eta_j(\epsilon)\frac{\spr{f_{j}^r(\epsilon)}{\P_j^r(\epsilon)}}{\setsize{ K_j(\epsilon)}}
-
\sum_{i=1}^{N(\delta)}\eta_i(\delta)\frac{\spr{f_{i}^r(\delta)}{\P_i^r(\delta)}}{\setsize{ K_i(\delta)}}
}.
\end{align*}
Our aim is to show $\lim_{\delta\searrow 0}\lim_{\epsilon\searrow 0}\mathcal D(\epsilon,\delta)=0$.
To prove this, we insert terms which interpolate between the minuend and the subtrahend.
These terms will be given using Theorem~\ref{thm:STP}.
For each $\epsilon\in\Ioc{1/(k+1)}{1/k}$,
we apply Theorem~\ref{thm:STP} to choose the sets $K_j(\epsilon)$,
$j=1,\dotsc,N(\ve)$, from the F{\o}lner sequence $(Q_{n+k})_{n\in\N}$.
The particular choice of the sets $K_j(\epsilon)$, $j=1,\dots,N(\epsilon)$,
as elements of the sequence $(Q_{n+k})_n$ ensures that for given $\delta>0$
we find $\epsilon_0>0$ such that for arbitrary $\epsilon\in\Ioo0{\ve_0}$
each $K_j(\epsilon)$, $j=1,\dots,N(\epsilon)$,
can be $\delta$-quasi tiled with the elements $K_i(\delta)$,
$i=1,\dots,N(\delta)$.
As in Theorem~\ref{thm:STP}, we denote the associated center sets by $T_i^j(\delta)$, where we emphasize the dependence on the parameter $\delta$.

For $K\in \cF$ we use the notation
\begin{align}\label{eq:defF}
F(K):=\spr{f_K}{\P_K}
\end{align}
and hence for the tiles $K_j(\epsilon)$, $i=1,\dots, N(\epsilon)$, we write $F(K_i^r(\epsilon)):=\spr{f_{i}^r(\epsilon)}{\P_i^r(\epsilon)}$.
The function~$F$ is translation invariant, i.\,e.\ for all $K\in\cF$ and $t\in G$ we have $F(K t)=F(K)$.

With the convention \eqref{eq:defF} and using the triangle inequality we obtain
$\mathcal D(\epsilon,\delta) \leq \mathcal D_1(\epsilon,\delta) +\mathcal D_2(\epsilon,\delta),$
where
\begin{align*}
\mathcal D_1(\epsilon,\delta)&:=\biggnorm{\sum_{j=1}^{N(\epsilon)}\eta_j(\epsilon)\frac{F(K_j^r(\epsilon))-\sum_{i=1}^{N(\delta)}\setsize{T_i^j(\delta)} F(K_i^r(\delta))}{\setsize{ K_j(\epsilon)}}
} \text{, and}\\
\cD_2(\epsilon,\delta)&:= \biggnorm{\sum_{j=1}^{N(\epsilon)}\eta_j(\epsilon)\frac{\sum_{i=1}^{N(\delta)}\setsize{T_i^j(\delta)} F(K_i^r(\delta))}{\setsize{ K_j(\epsilon)}}
-
\sum_{i=1}^{N(\delta)}\eta_i(\delta)\frac{F(K_i^r(\delta))}{\setsize{ K_i(\delta)}}
}.
\end{align*}
The translation invariance of~$F$ and the triangle inequality yield
\begin{equation}\label{eq:D1first}
  \mathcal D_1(\epsilon,\delta)
  \leq \sum_{j=1}^{N(\epsilon)}\frac{\eta_j(\epsilon)}{\setsize{ K_j(\epsilon)}}
    \biggnorm{F(K_j^r(\epsilon))-\sum_{i=1}^{N(\delta)}
      \sum_{t\in T_i^j(\delta)} F(K_i^r(\delta)t)}.
\end{equation}
We decompose $K_j^r(\epsilon)$ in the following way
\begin{multline*}
	K_j^r(\epsilon)
	=\Union_{i=1}^{N(\delta)}\Union_{t\in T_i^j(\delta)}K_i^r(\delta)t
	\quad \dunion \quad
	K_j^r(\ve)\setminus\Union_{i=1}^{N(\delta)}
		K_i(\delta)T_i^j(\delta)\quad\dunion\\\quad
	\dunion\Biggl(\biggl(K_j^r(\ve)\setminus\Union_{i=1}^{N(\delta)}
		K_i^r(\delta)T_i^j(\delta)
		\biggr)\isect
		\Union_{i=1}^{N(\delta)}
		\bigl(K_i(\delta)\isect\partial^r(K_i(\delta))\bigr)T_i^j(\delta)\Biggr)
		=:\alpha_1\dunion\alpha_2\dunion\alpha_3\text.
\end{multline*}
By definition of the function~$F$
the almost additivity of the admissible field~$f$ inherits to~$F$.
Note that $\delta$-disjointness of the sets $K_it$, $t\in T_i^j(\delta)$
implies $\delta$-disjointness of the sets $K_i^r t$, $t\in T_i^j(\delta)$.
Therefore, applying almost additivity, \cref{la:quasialmostadd}
and the properties of admissible fields and the boundary term we obtain
\begin{align*}&
	\biggnorm{F(K_j^r(\epsilon))-\sum_{i=1}^{N(\delta)}
		\sum_{t\in T_i^j(\delta)}F(K_i^r(\delta)t)}\\&
	\le\biggnorm{F(K_j^r(\epsilon))-\sum_{i=1}^3F(\alpha_i)}
		+\biggnorm{F(\alpha_1)-\sum_{i=1}^{N(\delta)}\sum_{t\in T_i^j(\delta)}
			F(K_i^r(\delta))}+\norm{F(\alpha_2)}+\norm{F(\alpha_3)}\\&
	\le\sum_{i=1}^3b(\alpha_i)+\delta(3\constf+9\constb)\setsize{K_j(\epsilon)}
		+3\sum_{i=1}^{N(\delta)}\sum_{t\in T_i^j}b(K_i^r(\delta))
		+\constf\setsize{\alpha_2}+\constf\setsize{\alpha_3}\\&
	\le\delta(3\constf+9\constb)\setsize{K_j(\epsilon)}
		+4\sum_{i=1}^{N(\delta)}\sum_{t\in T_i^j(\delta)}b(K_i^r(\delta))
		+(\constf+\constb)\setsize{\alpha_2}+(\constf+\constb)\setsize{\alpha_3}\text.
\end{align*}
Next, we estimate the sizes of~$\alpha_2$ and~$\alpha_3$.
For~$\alpha_3$ we drop some of the intersections in its definition.
In order to give a bound on the size of~$\alpha_2$, we use that~$K_j^r(\ve)$
is $(1-2\ve)$-covered by $\{K_i^r(\delta)\mid i\}$,
more specifically, part~\ref{qt:2epscover} in \cref{def:qt}.
We obtain
\begin{align*}
  \setsize{\alpha_2}
  \le2\delta\setsize{K_j(\epsilon)}
  \qquad\text{and}\qquad
  \setsize{\alpha_3}
  \le\sum_{i=1}^{N(\delta)}\setsize{T_i^j(\delta)}
	\setsize{\partial^r(K_i(\delta))}
	\text,
\end{align*}
and therewith achieve
\begin{align*}
& \biggnorm{F(K_j^r(\epsilon))-\sum_{i=1}^{N(\delta)}\sum_{t\in T_i^j(\delta)} F(K_i^r(\delta)t)}
\\&\leq
  \delta(5\constf+11\constb)\setsize{K_j(\epsilon)} + \sum_{i=1}^{N(\delta)}\setsize{ T_i^j(\delta)} \Bigl(4 b(K_i^r(\delta))+(\constf+\constb)\setsize{\partial^r(K_i(\delta))}\Bigr).
\end{align*}
This together with \eqref{eq:D1first} and part~\ref{etaNa} of Lemma~\ref{la:etaN} yields
\begin{align*}&
  \mathcal D_1(\epsilon,\delta)\\&
  \leq\sum_{j=1}^{N(\epsilon)}\!\biggl(\!\delta(5\constf+11\constb)\eta_j(\epsilon) + \sum_{i=1}^{N(\delta)}\!\frac{\eta_j(\epsilon)\setsize{ T_i^j(\delta)}}{\setsize{ K_j(\epsilon)}} \Bigl(4 b(K_i^r(\delta))+(\constf+\constb)\setsize{\partial^r(K_i(\delta))}\Bigr)\!\biggr)\\
&\leq \delta(5\constf+11\constb) +\sum_{j=1}^{N(\epsilon)}\sum_{i=1}^{N(\delta)}\frac{\eta_j(\epsilon)\setsize{ T_i^j(\delta)}}{\setsize{ K_j(\epsilon)}} \Bigl(4 b(K_i^r(\delta))+(\constf+\constb)\setsize{\partial^r(K_i(\delta))}\Bigr).
\end{align*}
As~$\delta$ is assumed to be smaller than~$1/10$,
we can apply \cref{cor:density},
which gives for arbitrary $i\in\{1,\dots,N(\delta)\}$
and $j\in \{1,\dots,N(\epsilon)\}$
\begin{align*}
  \frac{\setsize{T_i^j(\delta)}}{\setsize{K_j(\epsilon)}}
	\le\frac{\eta_i(\delta)}{\setsize{K_i(\delta)}}
		+4\frac{\delta \eta_i(\delta)}{\setsize{K_i(\delta)}}
	\le5\frac{\eta_i(\delta)}{\setsize{K_i(\delta)}}\text.
\end{align*}
Inserting this in the last estimate for $\cD_1(\epsilon,\delta)$ implies together with part~\ref{etaNa} of Lemma~\ref{la:etaN} that
\begin{align*}
 \mathcal D_1(\epsilon,\delta)
&\leq \delta(5\constf+11\constb) +\sum_{i=1}^{N(\delta)}\frac{5\eta_i(\delta)}{\setsize{ K_i(\delta)}} \Bigl(4 b(K_i^r(\delta))+(\constf+\constb)\setsize{\partial^r(K_i(\delta))}\Bigr).
\end{align*}
Now, we use the monotonicity assumption in \eqref{eq:monotone0}, which allows to replace the elements $K_i^r(\delta)$ and $K_i(\delta)$ by $Q_i^r$ and $Q_i$, respectively:
\begin{align}\label{eq:D1}
 \mathcal D_1(\epsilon,\delta)
&\leq \delta(5\constf+11\constb) +\sum_{i=1}^{N(\delta)}\frac{5\eta_i(\delta)}{\setsize{ Q_i}} \Bigl(4 b(Q_i^r)+(\constf+\constb)\setsize{\partial^r(Q_i)}\Bigr).
\end{align}

Let us proceed with the estimation of $\cD_2(\epsilon,\delta)$:
\begin{align}\label{eq:D2first}
 \cD_2(\epsilon,\delta)&= \biggnorm{\sum_{i=1}^{N(\delta)}F(K_i^r(\delta)) \biggl(\sum_{j=1}^{N(\epsilon)}\eta_j(\epsilon)\frac{\setsize{T_i^j(\delta)} }{\setsize{ K_j(\epsilon)}}
-
\frac{\eta_i(\delta)}{\setsize{ K_i(\delta)}}\biggr)}
.
\end{align}
With the triangle inequality, \cref{cor:density},
and part~\ref{etaNa} of \cref{la:etaN} we obtain
\begin{align*}
	\biggabs{\sum_{j=1}^{N(\epsilon)}
		\eta_j(\epsilon)\frac{\setsize{T_i^j(\delta)}}{\setsize{K_j(\epsilon)}}
		-\frac{\eta_i(\delta)}{\setsize{ K_i(\delta)}}}&
	\le\sum_{j=1}^{N(\epsilon)}\eta_j(\epsilon)
		\biggabs{\frac{\setsize{T_i^j(\delta)}}{\setsize{K_j(\epsilon)}}
		-\frac{\eta_i(\delta)}{\setsize{K_i(\delta)}}}
		+\biggabs{\sum_{j=1}^{N(\epsilon)}\eta_j(\epsilon)-1}
			\frac{\eta_i(\delta)}{\setsize{K_i(\delta)}}\\&
	\le\sum_{j=1}^{N(\epsilon)}\eta_j(\epsilon)
		\frac{4\delta\eta_i(\delta)}{\setsize{K_i(\delta)}}
		+\frac{\epsilon\eta_i(\delta)}{\setsize{K_i(\delta)}}
	\le\frac{4\delta\eta_i(\delta)}{\setsize{K_i(\delta)}}
		+\frac{\epsilon\eta_i(\delta)}{\setsize{ K_i(\delta)}}\text.
\end{align*}
This together with \eqref{eq:D2first} gives the bound
\begin{align}\label{eq:D2}
 \cD_2(\epsilon,\delta)
&\leq \sum_{i=1}^{N(\delta)} \constf \setsize{K_i^r(\delta)}\biggl(\frac{4\delta \eta_i(\delta)}{\setsize{K_i(\delta)}}
+
\frac{\epsilon \eta_i(\delta)}{\setsize{ K_i(\delta)}}\biggr)
\leq
4\constf\delta
+
\constf \epsilon .
\end{align}
Thus, the estimates of $\cD_1(\epsilon,\delta)$ and $\cD_2(\epsilon,\delta)$ in \eqref{eq:D1} and \eqref{eq:D2} together yield
\begin{align}\label{eq:cD}
  \cD(\epsilon,\delta)\leq
  \constf \epsilon +\delta(9\constf+11\constb) +\sum_{i=1}^{N(\delta)}\frac{5\eta_i(\delta)}{\setsize{ Q_i}} \bigl(4 b(Q_i^r)+(\constf+\constb)\setsize{\partial^r(Q_i)}\bigr)
\end{align}
for all $\delta>0$ and $\ve\in\Ioo0{\ve_0(\delta)}$.
Applying part~\ref{etaNc} of Lemma~\ref{la:etaN} we see
\begin{align*}
 \lim_{\delta\searrow 0}\lim_{\epsilon\searrow 0} \cD(\epsilon,\delta)=0\text.
\end{align*}
Using a Cauchy argument and the fact that $\B$ is a Banach space we obtain that there exists an element $f^*\in \B$ with
\[
  \lim_{\epsilon\searrow 0}
  \biggnorm{\sum_{j=1}^{N(\epsilon)}\eta_j(\epsilon)\frac{\spr{f_{j}^r(\epsilon)}{\P_j^r(\epsilon)}}{\setsize{ K_j(\epsilon)}}-f^*}=0.
\]
In order to get the error estimate for finite $\delta>0$,
we use \eqref{eq:cD}, \cref{la:etaN}\ref{etaNc},
and \eqref{eq:monotone0} as follows
\begin{align*}&
  \biggnorm{\sum_{j=1}^{N(\delta)}\eta_j(\delta)
    \frac{\spr{f_{j}^r(\delta)}{\P_j^r(\delta)}}{\setsize{ K_j(\delta)}}
    -f^*}
  =\lim_{\ve\dnto0}\cD(\ve,\delta)\\&
  \le(9\constf+11\constb)\delta
    +\sum_{i=1}^{N(\delta)}\frac{5\eta_i(\delta)}{\setsize{ Q_i}}
    \bigl(4b(Q_i^r)+(\constf+\constb)\setsize{\partial^r(Q_i)}\bigr)\\&
  \le(9\constf+11\constb)\delta+5(4+\constf+\constb)\beta(\delta)
    \text.\qedhere
\end{align*}
\end{proof}

\section{Proof of the main theorem}\label{sec:example}
We will prove a slightly more explicit statement which
tracks the geometric error in terms of~$\varepsilon$
and the probabilistic error in terms of~$\kappa$ separately.
\Cref{thm:main} is implied by the choice $\kappa:=\sqrt\varepsilon$.
Recall that~$\B$ is the Banach space of bounded and right-continuous functions
from~$\R$ to~$\R$.

\begin{Theorem}
  Let~$G$ be a finitely generated amenable group.
  Further, let~$\cA\in \cB(\R)$ and $(\Omega=\cA^G,\cB(\Omega),\P)$
  a probability space such that~$\P$ satisfies \labelcref{M1,M2,M3}.
  Finally, let~$\cU$ be an admissible set
  of admissible fields with common bound~$\constf[\cU]$,
  cf.\ \cref{def:admissible}.
  \par
  Then, there exists a limit element $f^*\in\B$ with the following properties.
  For each F{\o}lner sequence $(\Lambda_n)$, $\ve\in\Ioo0{1/10}$ and $\kappa>0$,
  there exist $j_0(\ve)\in\N$, which is independent of~$\kappa$
  and~$\constf[\cU]$,
  and $a(\ve,\kappa,\constf[\cU]),b(\ve,\kappa,\constf[\cU])>0$,
  such that for all $j\in\N$, $j\ge j_0(\ve)$,
  there is an event $\Omega_{j,\ve,\kappa,\constf[\cU]}\in\Borel(\Omega)$
  with the properties
  \begin{equation*}
    \P(\Omega_{j,\ve,\kappa,\constf[\cU]})
    \ge1-b(\ve,\kappa,\constf[\cU])
    \exp\bigl(-a(\ve,\kappa,\constf[\cU])\setsize{\Lambda_j}\bigr)
  \end{equation*}
  and
  \begin{align*}
    \biggnorm{\frac{f(\Lambda_j,\omega)}{\setsize{\Lambda_j}}-f^*}&
    \le(37\constf+47\constb+46)\sqrt\ve+\kappa
    \qtext{ for all $\omega\in\Omega_{j,\ve,\kappa,\constf[\cU]}$ and all $f\in\cU$.}
  \end{align*}
\end{Theorem}
\begin{proof}
  We follow the path prescribed in the previous chapters and
  \begin{itemize}[nosep]
  \item quasi tile~$\Lambda_j$, $j\ge j_0(\ve)$,
    with~$K_i(\varepsilon)$, $i=1,\dotsc,N(\varepsilon)$,
    see \cref{thm:STP},
  \item approximate $\setsize{\Lambda_j}^{-1}f(\Lambda_j,\omega)$
    with the empirical measures $L_{i,j}^{r,\omega}(\varepsilon)$,
    cf.~\eqref{def:Lijromegar} and \cref{la:quasi-first-step},
  \item express the empirical measures by their limiting counterparts~%
    $\P_i^r(\varepsilon)$ with \cref{la:gc-main-lemma}, and
  \item use the Cauchy property of the remaining terms
    to obtain a limiting function~$f^*$, see \cref{la:cauchy}.
  \end{itemize}
  To confirm the error estimate, we employ the triangle inequality
  \begin{align*}
    \biggnorm{\frac{f(\Lambda_j,\omega)}{\setsize{\Lambda_j}}-f^*}&
    \le\biggnorm{\frac{f(\Lambda_j,\omega)}{|\Lambda_j|}-
      \sum_{i=1}^{N(\epsilon)}\eta_i(\epsilon)
        \frac{\spr{f_{i}^r(\epsilon)}{L^{r,\omega}_{i,j}(\epsilon)}}
             {\setsize{ K_i(\epsilon)}}
    }\\&\quad
    +\biggnorm{
      \sum_{i=1}^{N(\epsilon)}\eta_i(\epsilon)
        \frac{\spr{f_{i}^r(\epsilon)}{L^{r,\omega}_{i,j}(\epsilon)}}
             {\setsize{ K_i(\epsilon)}}
      -\sum_{i=1}^{N(\epsilon)}\eta_i(\epsilon)
        \frac{\spr{f_{i}^r(\epsilon)}{\P_i^r(\epsilon)}}
             {\setsize{ K_i(\epsilon)}}
    }\\&\quad
    +\biggnorm{\sum_{i=1}^{N(\epsilon)}\eta_i(\epsilon)
      \frac{\spr{f_{i}^r(\epsilon)}{\P_i^r(\epsilon)}}
           {\setsize{ K_i(\epsilon)}}-f^*}
    =:\Delta(\ve,j,\omega)\text.
  \end{align*}
  By \cref{la:quasi-first-step,la:gc-main-lemma,la:cauchy},
  we immediately get that there is an event $\tilde\Omega\in\Borel(\Omega)$
  with full probability $\P(\tilde\Omega)=1$ such that
  $\lim_{\ve\dnto0}\lim_{j\to\infty}\Delta(\ve,j,\omega)=0$
  for all $\omega\in\tilde\Omega$.
  Furthermore, \cref{la:gc-main-lemma} provides the event
  $\Omega_{j,\ve,\kappa,\constf[\cU]}$ with probability as large as claimed,
  and by collecting all the error terms and by \cref{rem:beta_le_ve},
  we see that for all $\ve\in\Ioo0{1/10}$, $j\ge j_0(\ve)$, $\kappa>0$,
  $f\in\cU$, and $\omega\in\Omega_{j,\ve,\kappa,\constf[\cU]}$,
  see \cref{la:gc-main-lemma},
  \begin{align*}
    \biggnorm{\frac{f(\Lambda_j,\omega)}{\setsize{\Lambda_j}}-f^*}&
    \le
    (20\constf+30\constb)\ve+(17\constf+17\constb+46)\beta(\ve)+\kappa\\&
    \le(37\constf+47\constb+46)\sqrt\ve+\kappa
    \text.
  \end{align*}
  Note the uniformity of the last inequality for all $f\in \cU$ is also discussed in Remark \ref{rem:beta_le_ve}.

  To see that the limit $f^*$ does not depend on the specific choice of $(\Lambda_j)$ use the following argument: Every two F{\o}lner sequences can be combined two one F{\o}lner sequence, which yields by our theory a limit $f^*\in\mathbb B$. As the two original sequences are subsequences, they lead to the same limit function $f^*$.
\end{proof}

\appendix
\small
\section{Conditional resampling}

In \cref{lemma:independence},
we need to remove the dependent parts of samples.
We achieve this by resampling the critical parts of the samples,
keeping the large enough already independent parts.
This is done by augmenting the probability space
to provide room for more random variables.
The problem of resampling turned out to be treatable in a much broader setting,
so a general tool is provided here.

\begin{Theorem}[Resampling]\label{thm:resampling}
  Let $(\Omega,\cA,\P)$ be a Borel probability space, $(S,\cS)$ a Borel space,
  and $X\from\Omega\to S$ an $S$-valued random variable with distribution
  $\P_X:=\P\circ X^{-1}\from\cS\to\Icc01$.
  Further let~$I$ be an index set, and for each $j\in I$,
  let $\cY_j\subseteq\cS$ be a $\sigma$-algebra.
  \par
  Then, there is a probability space $(\superOmega,\superA,\superP)$
  such that for all $j\in I$, maps as indicated in the following diagram
  exist and are measure preserving, and all the diagrams commute almost surely.
  \begin{center}
  \begin{tikzpicture}
    \node (superOmega) {$(\superOmega,\superA,\superP)$};
    \node (Yj) [below=of superOmega] {$(S,\cY_j,\P_X|_{\cY_j})$};
    \node (Sj) [right=of Yj] {$(S,\cS,\P_X)$};
    \node (S) [left=of Yj] {$(S,\cS,\P_X)$};
    \node (Omega) [above=of S] {$(\Omega,\cA,\P)$};
    \draw (Omega) edge [->,very thick] node [auto,swap] {$X$} (S);
    \draw (superOmega)
      edge [->,thick] node [auto] {$\Pi_0$} (Omega)
      edge [->,thick] node [auto] {$\superX$} (S)
      edge [->,thick] node [auto] {$X_j$} (Sj);
    \draw (Yj)
      edge [<-,thick] node [auto] {$\id_S$} (S)
      edge [<-,thick] node [auto,swap] {$\id_S$} (Sj);
  \end{tikzpicture}
  \end{center}
  This means in particular that~$\Pi_0$ is measure preserving,
  and that, for all $j\in I$,
  \begin{enumerate}[(i), noitemsep]
    \item\label{resampling:distribution}
      the random variable~$X_j$ has distribution~$\P_X$,
    \item\label{resampling:equality}
      for each measure space $(T,\cT)$ and each
      $\cY_j$-$\cT$-measurable map $g\from(S,\cY_j)\to(T,\cT)$,
      we have $g(\superX)=g(X_j)$ $\superP$-almost surely.
  \end{enumerate}
  Furthermore, the joint distribution of $(X_j)_{j\in I}$
  has the following properties.
  \begin{enumerate}[(i), noitemsep, resume]
    \item\label{resampling:condinde}
      For each finite subset $F\subseteq I$
      and $A_F=\bigtimes_{j\in F}A_j$,
      where $A_j\in\cS$, we have $\P_X$-almost surely that
      \begin{equation*}
        \superP(X_F\in A_F\given\superX=\argmt)
        =\prod\nolimits_{j\in F}\superP(X_j\in A_j\given\superX=\argmt)
        =\prod\nolimits_{j\in F}\P_X(A_j\given\cY_j)\text.
      \end{equation*}
      In particular, the random variables~$X_j$, $j\in I$,
      are independent when conditioned on~$\superX$.
    \item\label{resampling:independ}
      If, for a (not necessarily finite) subset $J\subseteq I$,
      the $\sigma$-algebras $\cY_j$, $j\in J$, are $\P_X$-independent,
      then the random variables~$X_j$, $j\in J$, are $\superP$-independent.
  \end{enumerate}
\end{Theorem}

Since~$\Pi_0$ is measure preserving, $(\superOmega,\superA,\superP)$
extends $(\Omega,\cA,\P)$.
Property~\ref{resampling:distribution} justifies the name resampling.
Statement~\ref{resampling:equality}
says that in~$X_j$ the information contained in~$\cY_j$
is preserved throughout the resampling, $j\in I$.
Point~\ref{resampling:condinde} states that the new random variables
copied only the information from~$\cY_j$, $j\in I$, and not more.
In~\ref{resampling:independ},
we learn how to provide independence of the resampling random variables.

\begin{proof}
  We define the spaces and maps as follows:
  \begin{align*}
    \superOmega:=\Omega\times S^I&\textq,
    \superA:=\cA\tensor\cS^{\tensor I}
    \text,\\
    \Pi_0\from\superOmega\to\Omega&\textq,
    \Pi_0(\omega,(s_j)_{j\in I}):=\omega
    \text,\\
    \superX\from\superOmega\to S&\textq,
    \superX(\omega,(s_j)_{j\in I}):=X(\omega)
    \text,\\
    X_j\from\superOmega\to S&\textq,
    X_j(\omega,(s_{k})_{k\in I}):=s_j
  \end{align*}
  We now define the measure~$\superP$ via Kolmogorov's extension theorem,
  see \cite[Theorem~14.36]{Klenke2008}.
  We need a consistent family of probability measures.
  For a more unifying notation, we augment $I_0:=\{0\}\dunion I$.
  Fix a finite subset $F\subseteq I_0$.
  If $0\in F$, we define a probability measure
  $\P^F\from\cA\tensor\cS^{\tensor F\setminus\{0\}}\to\Icc01$.
  In case $0\notin F$, we define a probability measure
  $\P^F\from\cS^{\tensor F}\to\Icc01$.
  If $0\in F$, then choose $A_0\in\cA$,
  otherwise, let $A_0:=\Omega$.
  For all $j\in F\setminus\{0\}$ we let $A_j\in\cS$.
  Now let $A_F:=\bigtimes_{j\in F}A_j$ and
  \begin{equation}\label{eq:defPF}
    \P^F(A_F)
    :=\E\Bigl[\ifu{A_0}\prod\nolimits_{j\in F\setminus\{0\}}
      \P_X(A_j\given\cY_j)\circ X\Bigr]
    \text.
  \end{equation}
  Here, $\E$ denotes integration with respect to~$\P$.
  By the extension theorem for measures, see \cite[Theorem~1.53]{Klenke2008},
  \eqref{eq:defPF} defines a probability measure.
  The family $(\P^F)_{\text{$F\subseteq I$ finite}}$ is consistent.
  For example, for finite subsets $0\notin F\subseteq J\subseteq I$
  with the projection $\Pi^J_F\from S^J\to S^F$
  and $A_F=\bigtimes_{j\in F}A_j$ with $A_j\in\cS$,
  we have $(\Pi^J_F)^{-1}(A_F)=A_F\times\bigtimes_{j\in J\setminus F}S$.
  Thus,
  \begin{equation*}
    \P^J\bigl((\Pi^J_F)^{-1}(A_F)\bigr)
    =\E_X\Bigl[\prod\nolimits_{j\in F}\P_X(A_j\given\cY_j)
      \prod\nolimits_{j\in J\setminus F}\P_X(S\given\cY_j)\Bigr]
    =\P^F(A_F)\text,
  \end{equation*}
  where~$\E_X$ is integration with respect to~$\P_X$.
  The remaining cases $0\in F\subseteq J$, and $0\notin F$ but $0\in J$
  work analogously.
  By Kolmogorov's extension theorem, we have exactly one measure
  $\superP:=\varprojlim_{F\subseteq I}\P^F\from\superA\to\Icc01$.
  \par
  We now verify the properties of~$\superP$.
  Let us first check, that~$\Pi_0$ is measure preserving.
  Indeed, for $A\in\cA$, we have
  \begin{equation*}
    \superP(\Pi_0\in A)
    =\P^{\{0\}}(A)
    =\E[\ifu A]
    =\P(A)\text.
  \end{equation*}
  Now we already know that $\superX=X\circ\Pi_0$
  is measure preserving, too.
  \begin{enumerate}[wide]
  \item[Ad \ref{resampling:distribution}:]
    For all $j\in I$ and $B\in\cS$, we have
    \begin{equation*}
      \superP(X_j\in B)
      =\P^{\{j\}}(B)
      =\E_X[\P_X(B\given\cY_j)]
      =\E_X[\ifu B]
      =\P_X(B)\text.
    \end{equation*}
  \item[Ad \ref{resampling:equality}:]
    Let $j\in I$, $(T,\cT)$ be a measure space and $g\from S\to T$
    be $\cY_j$-$\cT$-measurable.
    We determine the joint distribution of~$X$ and~$X_j$.
    By~\eqref{eq:defPF}, we have, for $B,B'\in\cT$,
    that $A:=g^{-1}(B)\in\cY_j$ as well as $A':=g^{-1}(B')\in\cY_j$, and
    \begin{align}
      \superP(g(\superX)\in B,g(X_j)\in B')&
      =\superP(\superX\in A,X_j\in A')
      =\P^{\{0,j\}}(X^{-1}(A)\times A')\notag\\&
      =\E[\ifu{X^{-1}(A)}\P_X(A'\given\cY_j)\circ X]
      =\E_X[\ifu A\ifu{A'}]\notag\\&
      =\P_X(A\isect A')
      =\superP(\superX\in A\isect A')
      =\superP(g(\superX)\in B\isect B')\text,\label{eq:isectmeas}
    \end{align}
    where in the last line, we used that
    $A\isect A'=g^{-1}(B)\isect g^{-1}(B')=g^{-1}(B\isect B')$.
    Now, since the rectangles $\{B\times B'\mid B,B'\in\cT\}$
    are stable under intersections and generate $\cT\tensor\cT$,
    equation~\eqref{eq:isectmeas} determines the distribution of
    $(g(\superX),g(X_j))\from\superOmega\to T^2$.
    Note, that the measure which is concentrated on the diagonal
    $\{(t,t)\mid t\in T\}$ with both marginals equal to
    $\P_X\circ g^{-1}$ satisfies~\eqref{eq:isectmeas}, too.
    Therefore, $\superP(g(\superX)=g(X_j))=1$.
  \item[Ad \ref{resampling:condinde}:]
    Fix a finite subset $F\subseteq I$ and $A_j\in\cS$ for $j\in F$,
    and let $A_F:=\bigtimes_{j\in F}A_j$.
    For all $B\in\cS$, we have
    \begin{align*}
      \superE[\ifu{\{\superX\in B\}}\superP(X_F\in A_F\given\superX)]&
      =\superE[\ifu{\{\superX\in B\}}\superE[\ifu{\{X_F\in A_F\}}\given\superX]]
      =\superE[\ifu{\{\superX\in B\}}\ifu{\{X_F\in A_F\}}]\\&
      =\superP[\superX\in B,X_F\in A_F]
      =\P^{\{0\}\union F}(X^{-1}(B)\times A_F)\\&
      =\E\Bigl[\ifu{X^{-1}(B)}\prod\nolimits_{j\in F}
        \P_X(A_F\given\cY_j)\circ X\Bigr]\\&
      =\superE\Bigl[\ifu{\{\superX\in B\}}\prod\nolimits_{j\in F}
        \P_X(A_F\given\cY_j)\circ\superX\Bigr]\text.
    \end{align*}
    Since $\sigma(\superX)=\{\{\superX\in B\}\mid B\in\cS\}$, this proves
    \begin{equation*}
      \superP(X_F\in A_F\given\superX)
      =\prod\nolimits_{j\in F}\P_X(X_j\in A_j\given\cY_j)\circ\superX
    \end{equation*}
    $\superP$-almost surely.
    For $F=\{j\}$, we get
    $\superP(X_j\in A_j\given\superX)=\P_X(X_j\in A_j\given\cY_j)$, too.
    The claim is the factorized version of these statements,
    which exist because $(S,\cS)$ is a Borel space.
  \item[Ad \ref{resampling:independ}:]
    For $F\subseteq J$ finite and $A_F=\bigtimes_{j\in F}A_j$
    with $A_j\in\cS$, we use \ref{resampling:condinde} to get
    \begin{align*}
      \superP(X_F\in A_F)&
      =\superE[\superP(X_F\in A_F\given\superX)]\\&
      =\superE\Bigl[\prod\nolimits_{j\in F}
        \P_X(A_j\given\cY_j)\circ\superX\Bigr]
      =\E_X\Bigl[\prod\nolimits_{j\in F}\P_X(A_j\given\cY_j)\Bigr]
        \text.
    \end{align*}
    The $\sigma$-algebras~$\cY_j$, $j\in F\subseteq J$,
    are $\P_X$-independent.
    This independence is inherited by
    $\cY_j$-measurable functions like $\P_X(A_j\given\cY_j)$.
    We can therefore continue the calculation with
    \begin{align*}
      \superP(X_F\in A_F)&
      =\prod\nolimits_{j\in F}\E_X\bigl[\P_X(A_j\given\cY_j)\bigr]
      =\prod\nolimits_{j\in F}\P_X(A_j)
      =\prod\nolimits_{j\in F}\superP(X_j\in A_j)\text.
    \end{align*}
    Since the cylinder sets generate $\cS^{\tensor J}$,
    this is the claimed $\superP$-independence.\qedhere
  \end{enumerate}
\end{proof}

\section{Proof summary for montilable amenable groups}\label{monotile}

The proofs of \cite{SchumacherSV-16} concerning the case $G=\Z^d$ can be generalized to apply to a
finitely generated amenable group~$G$ if it satisfies the tiling property~\labelcref{tiling}.

We list the major changes which are necessary for this purpose:
\begin{enumerate}[(a)]
  \item\label{pointa}
    Instead of defining the set~$T_{m,n}$ using multiples of~$m$
    (c.\,f.~eq.~(3.1) in \cite{SchumacherSV-16}),
    we employ the grid~$T_m$, namely we set
    \begin{align}\label{def:Tmn}
      T_{m,n}:=\{t\in T_m\mid \Lambda_m t\subseteq \Lambda_n\}
    \end{align}
    Thus, $T_{m,n}$ contains the elements of~$T_m$
    which correspond to translates of~$\Lambda_m$
    which are completely contained in~$\Lambda_n$.
    Using this definition, the empirical measures are
    $L_{m,n}^\omega$ and $L_{m,n}^{\omega,r}$ are given accordingly.
  \item\label{pointb}
    One needs to verify the following basic result.
    Given a tiling F{\o}lner sequence $(\Lambda_n)$, we have
    \begin{enumerate}[(i),noitemsep]
      \item for each $m\in\N$ the sequence $(\Lambda_mT_{m,n})_{n\in\N}$
        is a F{\o}lner sequence;
      \item for each $m,n\in\N$ we have
        $\Lambda_n\subseteq\partial^{\rho(m)}(\Lambda_n)\union\Lambda_m T_{m,n}$,
        where $\rho(m)=\diam(\Lambda_m)$; and
      \item for each $m\in\N$ we have
        $\lim_{n\to\infty}\setsize{\Lambda_n}/\setsize{T_{m,n}}=\setsize{\Lambda_m}$.
    \end{enumerate}
  \item Points \labelcref{pointa,pointb} allow to prove
    an equivalent version of Lemma 3.2 of \cite{SchumacherSV-16}
    in the situation of amenable groups with property~\labelcref{tiling},
    by following exactly the steps of the proof presented therein.
  \item Besides Lemma 3.2.~also Lemma 6.1.~needs to be slightly changed.
    In fact, again by using \labelcref{pointa,pointb}
    the proof can directly be adapted to the situation where~$G$
    is amenable and $(\Lambda_n)$ is a tiling F{\o}lner sequence.
  \item In the end, the proof of the main theorem reduces basically
    to an application of the triangle inequality,
    the new versions of Lemma 3.2 and Lemma 6.1 as well as
    (the original version of) Theorem 5.1.
    Note that Theorem 5.1 need not to be adapted
    as it is independent of the geometry.
\end{enumerate}



\end{document}